\documentclass[journal,12pt,onecolumn,draftclsnofoot]{IEEEtran}
\usepackage{cite}
\usepackage{url}
\usepackage{mathrsfs}
\usepackage{amsfonts,amssymb}
\usepackage{bm}
\usepackage{amsmath}
\usepackage{amssymb}
\usepackage{breqn}
\usepackage[ruled,linesnumbered]{algorithm2e}
\usepackage{graphicx}
\usepackage{subfigure}
\usepackage{stfloats}
\usepackage{epstopdf}
\usepackage{amsfonts}
\usepackage{bbm}
\usepackage{amsthm}
\usepackage{float}
\usepackage{diagbox}
\usepackage{enumitem}
\usepackage{color}
\newtheorem{theorem}{Theorem}
\newtheorem{definition}{Definition}

\newtheorem{lemma}{Lemma}

\renewcommand\IEEEkeywordsname{Keywords}

\ifodd 1

\else

\fi

\begin{document}
\title{Computation Peer Offloading for Energy-Constrained Mobile Edge Computing in Small-Cell Networks}

\author{Lixing~Chen,~\IEEEmembership{Student~Member,~IEEE,}\\
        Sheng~Zhou,~\IEEEmembership{Member,~IEEE,}
        Jie~Xu,~\IEEEmembership{Member,~IEEE}

\thanks{L. Chen and J. Xu are with the Department of Electrical and
	Computer Engineering, University of Miami, USA. Email: lx.chen@miami.edu, jiexu@miami.edu. S. Zhou is with the Department of Electronic Engineering, Tsinghua University, China. Email: sheng.zhou@tsinghua.edu.cn.}
}

\maketitle

\begin{abstract}
The (ultra-)dense deployment of small-cell base stations (SBSs) endowed with cloud-like computing functionalities paves the way for pervasive mobile edge computing (MEC), enabling ultra-low latency and location-awareness for a variety of emerging mobile applications and the Internet of Things. To handle spatially uneven computation workloads in the network, cooperation among SBSs via workload peer offloading is essential to avoid large computation latency at overloaded SBSs and provide high quality of service to end users. However, performing effective peer offloading faces many unique challenges due to limited energy resources committed by self-interested SBS owners, uncertainties in the system dynamics and co-provisioning of radio access and computing services. This paper develops a novel online SBS peer offloading framework, called OPEN, by leveraging the Lyapunov technique, in order to maximize the long-term system performance while keeping the energy consumption of SBSs below individual long-term constraints. OPEN works online without requiring information about future system dynamics, yet provides provably near-optimal performance compared to the oracle solution that has the complete future information. In addition, this paper formulates a peer offloading game among SBSs, analyzes its equilibrium and efficiency loss in terms of the price of anarchy to thoroughly understand SBSs' strategic behaviors, thereby enabling decentralized and autonomous peer offloading decision making. Extensive simulations are carried out and show that peer offloading among SBSs dramatically improves the edge computing performance.

\end{abstract}
\section{Introduction}
Pervasive mobile devices and the Internet of Things are driving the development of many new applications, turning data and information into actions that create new capabilities, richer experiences and unprecedented economic opportunities. Although cloud computing enables convenient access to a centralized pool of configurable and powerful computing resources, it often cannot meet the stringent requirements of latency-sensitive applications due to the often unpredictable network latency and expensive bandwidth \cite{mao2017mobile,shi2016edge,roman2016mobile}. The growing amount of distributed data further makes it impractical or resource-prohibitive to transport all the data over today's already-congested backbone networks to the remote cloud \cite{rivera2014gartner}. As a remedy to these limitations, mobile edge computing (MEC) \cite{mao2017mobile,shi2016edge,roman2016mobile} has recently emerged as a new computing paradigm to enable in-situ data processing at the network edge, in close proximity to mobile devices and connected things. Located often just one wireless hop away from the data source, edge computing provides a low-latency offloading infrastructure, and an optimal site for aggregating, analyzing and distilling bandwidth-hungry data from end devices.
\begin{figure}[htb]
	\centering	
	\includegraphics [width=0.5\linewidth]{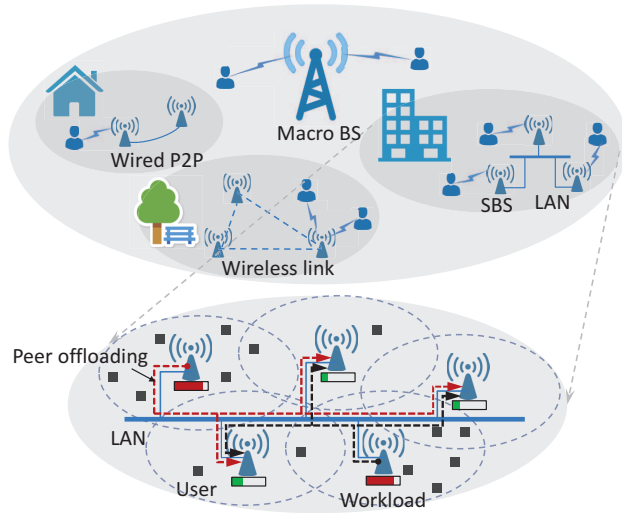}
	\caption{Illustration of SBS peer offloading.}
	\label{illustration}
	\vspace{-0.2 in}
\end{figure}

Considered as a key enabler of MEC, small-cell base stations (SBSs), such as femtocells and picocells, endowed with cloud-like computing and storage capabilities can serve end users' computation requests as a substitute of the cloud \cite{roman2016mobile}. Nonetheless, compared to mega-scale data centers, SBSs are limited in their computing resources. Since the computation workload arrivals in small cell networks can be highly dynamic and heterogeneous, it is very difficult for an individual SBS to provide satisfactory computation service at all times. To overcome these difficulties, cooperation among SBSs can be exploited to enhance MEC performance and improve the efficiency of system resource utilization via computation peer offloading. For instance, a cluster of SBSs can coordinate among themselves to serve mobile users by offloading computation workload from SBSs located in hot spot areas to nearby peer SBSs with light computation workload, thereby balancing workload among the geographically distributed SBSs (see Figure \ref{illustration} for an illustration). Similar ideas have been investigated for data-center networks to deal with spatial diversities of workload patterns, temperatures, and electricity prices. In fact, SBS networks are more vulnerable to heterogeneous workload patterns than data center networks which serve an aggregation of computation requests across large physical regions. Since the serving area of each SBS is small, the workload pattern can be affected by many factors such as location, time, and user mobility, therefore becoming very violate and easily leading to uneven workload distribution among the SBSs. Although there have been quite a few works on geographical load balancing in data centers, performing peer offloading in MEC-enabled small cell networks faces unique challenges.

First, small cells are often owned and deployed by individual users. Although incentive mechanism design, which has been widely studied in the literature for systems not limited to small cell networks, plays an important role in incentivizing self-interested users to participate in the collaboration of workload peer offloading, an equally, if not more, important problem is how to maximize the value of the limited resources committed by individual SBS owners. Second, small cells operate in a highly stochastic environment with random workload arrivals in both temporal and spatial domains. As a result, the long-term system performance is more relevant than the immediate performance. However, the limited energy resources committed by the SBS owners make the peer offloading decisions across time intricately intertwined, yet the decisions have to be made without foreseeing the far future. Third, whereas data centers manage only the computing resources, moving the computing resources to the network edge leads to the co-provisioning of radio access and computing services by the SBSs, thus mandating a new model for understanding the interplay and interdependency between the management of the two resources under energy constraints.

In this paper, we study computation peer offloading in MEC-enabled small cell networks. Our goal is to maximize the long-term system-wide performance (i.e. minimizing latency) while taking into account the limited energy resources committed by individual SBS owners. The main contributions of this paper are summarized as follows:

 1) We develop a novel framework called OPEN (which stands for Online PEer OffloadiNg) for performing stochastic computation peer offloading among a network of MEC-enabled SBSs in an online fashion by leveraging the Lyapunov optimization \cite{neely2010stochastic}. We prove that OPEN achieves within a bounded deviation from the optimal system performance that can be achieved by an oracle algorithm that knows the complete future information, while bounding the potential violation of the energy constraints imposed by individual SBS owners.
 
 2) We theoretically characterize the optimal peer offloading strategy. We show that the peer offloading decisions are determined by the \emph{marginal computation cost} (MaCC) -- a critical quantity that captures both computation delay cost and energy cost at SBSs. The peer offloading essentially is to evenly distribute MaCCs among SBSs. The SBSs decide their roles (to send or receive workload) based on the pre-offloading MaCCs (i.e., MaCCs before peer offloading). The amount of workload to be offloaded is determined based on optimal post-offloading MaCCs (i.e., MaCCs to achieve after peer offloading) designed by OPEN.

3) We consider both the scenario in which a central entity (e.g. the network operator) collects all current time information and coordinates the peer offloading and the scenario in which SBSs coordinate their peer offloading strategies in a decentralized and autonomous way. For the latter case, we formulate a novel peer offloading game, prove the existence of a Nash equilibrium using the variational inequality technique, and characterize the efficiency loss due to the strategic behaviors of SBSs in terms of the \emph{price of anarchy} (PoA).

4) We run extensive simulations to evaluate the performance of OPEN and verify our analytical results for various system configurations and traffic arrival patterns. The results confirm that our method significantly improves the system performance in terms of latency reduction and energy efficiency.

The rest of this paper is organized as follows. Section \ref{sec_related_work} reviews related works. Section \ref{sec_sys_model} presents the system model and formulates the problem. Section \ref{sec_online} develops the OPEN framework and presents the centralized solution for computation peer offloading. Section \ref{sec_noncoop} formulates and analyzes the peer offloading game. Simulations are carried out in Section \ref{sec_simulation}, followed by the conclusion in Section \ref{sec_conclusion}.

\section{Related Work}\label{sec_related_work}

\begin{table*}[htb]
	\centering
	\caption{Comparison with existing works}
	\begin{tabular}{lccccccc}
		\hline
		\diagbox{Feature}{Approach} & \cite{abdelnasser2014clustering,Guruacharya2013Dynamic} &  \cite{sardellitti2015joint,queis2015small,queis2015fogbalancing} &\cite{islam2015water} & \cite{liu2015greening} & \cite{zhang2013dynamic} & \cite{chen2016efficient,chen2015decentralized} & OPEN (This paper) \\
		\hline
		Applied stage               & UE-to-ES & UE-to-ES &  DC-to-DC    & DC-to-DC & UE-to-DC   & UE-to-ES & ES-to-ES\\
		Radio access aware      & Yes          & Yes          &  No               & No            & No              & Yes           & Yes\\
		Computation aware      & No           & Yes          & Yes               & Yes           & Yes             & Yes           & Yes\\
		System objective          & Myopic    & Myopic     &  Long-term   & Myopic      & Long-term  &  Myopic    &Long-term\\
		Long-term constraints & No           & No           &  Yes (Overall)& No            & No               &    No        &Yes (Individual)\\
		Temporal correlation    & No           & No           &  Yes              & No            & Yes              &   No         & Yes\\
		Strategic behavior        & No           & No          &  No                & No            & Yes              & Yes          & Yes\\
		\hline
	\vspace{-0.05 in}	
	\end{tabular}\\
	\hspace{-3.4 in}
	\vspace{-0.2 in}
	\begin{footnotesize} UE: User equipment; DC: Data Center; ES: Edge Server \end{footnotesize}
\end{table*}

The concept of offloading data and computation in cloud computing is used to address the inherent problems in mobile computing by using resource providers other than the mobile device itself to host the execution of mobile applications \cite{fernando2013mobile}. In the most common case, mobile cloud computing means to run an application on a resource rich cloud server located in remote mega-scale data centers, while the mobile device acts like a thin client connecting over to the remote server through 4G/Internet \cite{satyanarayanan2010mobile}. Recently, the edge computing paradigm \cite{shi2016edge} (a.k.a. fog computing \cite{bonomi2012fog}, cloudlet \cite{satyanarayanan2009cloudlet}, micro datacenter \cite{greenberg2008cost}) brings computing resources closer to the end users to enable ultra-low latency and precise location-awareness, thereby supporting a variety of emerging mobile applications such as mobile gaming, augmented reality and autonomous vehicles. Nevertheless, edge servers, such as MEC-enabled SBSs \cite{TROPIC}, cannot offer the same computation and storage capacities as traditional computing servers.

Many recent works investigate SBS cooperation for improving the system performance, subject to various constraints including local resource availability (e.g. radio resources\cite{abdelnasser2014clustering,Guruacharya2013Dynamic}, computational capacities \cite{queis2015fogbalancing}, energy consumption budgets\cite{rubio2014association} and backhaul bandwidth capacity  \cite{tam2017joint}). However, most of these works focus on optimizing the radio access performance only without considering the computing capability of SBSs.
In \cite{queis2015small,queis2015fogbalancing}, computation load distribution among the network of SBSs is investigated by considering both radio and computational resource constraints. Clustering algorithms are proposed to maximize users' satisfaction ratio while keeping the communication power consumption low. However, these works focus more on the user-to-SBS offloading side whereas our paper studies the offloading among peer SBSs. More importantly, these works perform myopic optimization without considering the stochastic nature of the system whereas our paper studies a problem that is highly coupled across time due to the long-term energy constraints.

Computation workload peer offloading among SBSs is closely related to geographical load balancing techniques originally proposed for data centers to deal with spatial diversities of workload patterns \cite{lin2012online}, temperatures\cite{xu2015temperature}, and electricity prices\cite{lou2015spatio}. Most of these works study load balancing problems that are independent across time \cite{liu2011greening}. Very few works consider temporally coupled problems. In \cite{lin2012online}, the temporal dependency is due to the switching costs (turning on/off) of data center servers, which significantly differs from our considered problem. The closest work to our paper is \cite{islam2015water}, which aims to minimize the long-term operational cost of data centers subject to a long-term water consumption constraint. However, the long-term constraint is imposed on the entire system whereas in our paper each SBS has an individual energy budget constraint. Moreover, we not only provide centralized solutions for peer-offloading but also develop schemes that enable autonomous coordination among SBSs by formulating and studying a peer-offloading game. 
Several works use reinforcement learning to efficiently manage the resource of geo-distributed data centers \cite{farahnakian2014energy,zhou2017reinforcement}. Although the reinforcement learning can also be a potential solution to our peer-offloading problem, there are several challenges for using such a formulation. First, the system states are usually assumed to be Markovian, which may not be true in real systems. Second, large state and action spaces may be needed to capture the various system and decision variables and hence complexity and convergence is a big issue. Third, reinforcement learning often does not capture the long-term energy constraint and may easily violate it. By contrast, our solution is based on the \emph{Lyapunov drift-plus-penalty} framework, which can be applied to more general stochastic systems, does not need to maintain large state and action spaces, and can handle long-term energy constraints.  

The formulated peer offloading game is similar to the widely studied congestion game \cite{nisan2007algorithmic} at the first sight. However, there is a crucial difference between these two games: a main assumption in congestion games is that all players have the same cost function for a element; however, in the peer offloading game, cost function of a SBS to retain workload on itself is different from that of other SBSs to offload tasks to that SBS due to the energy consumption concern. This difference demands for new analytical tools for understanding the peer offloading game. For instance, the potential function technique \cite{nisan2007algorithmic} used to establish the existence of a Nash equilibrium in the congestion game does not apply and hence, in this paper, we prove the existence of a Nash equilibrium via the variational inequality technique \cite{kinderlehrer2000introduction}. Game theoretic modeling was also applied in the MEC computation offloading context in \cite{chen2016efficient}\cite{chen2015decentralized}. This work focuses on the computation offloading among multiple UEs to a single BS, which is a different scenario than ours. Table 1 summarizes the differences of proposed strategy from existing works.

\section{System Model}\label{sec_sys_model}

\subsection{Network model}
We consider $N$ SBSs (e.g. femtocells), indexed by $\mathcal{N} = \{1,\dots, N\}$, deployed in a building (residential or enterprise) and connected by the same Local Area Network (LAN). These SBSs are endowed with, albeit limited, edge computing capabilities and hence, User Equipments (UEs) can offload their computation tasks to corresponding serving SBSs via wireless communications for processing. The computing capabilities of SBS $i$ is characterized by its computation service rate $f_i$ (CPU frequency), and the computation service rates of all SBSs in the network are collected by $\bm{f} = \{f_i\}_{i\in\mathcal{N}}$. Let $\mathcal{M}=\{1,\dots,M\}$ denote the set of all UEs in the building. Each SBS serves a dedicated set of UEs in its serving area, denoted by $\mathcal{M}_i\subseteq \mathcal{M}$. For example, UEs (e.g. mobile phones, laptops etc.) of employees in a business are authorized to access the communication/computing service of the SBS deployed by the business. Notice that our algorithm is also compatible with other network structure and association strategies as long as the UE-SBS associations stay unchanged in one peer offloading decision cycle. 

\subsection{Workload arrival model}
The operational timeline is discretized into time slots (e.g. 1 - 5 minutes) for making peer offloading decisions, which is a much slower time scale than that of task arrivals. In each time slot $t$, computation tasks originating from UE $m$ is generated according to a Poisson process which is a common assumption on the computation task arrival in edge systems \cite{mao2017mobile}.  Let $\pi^t_m$ denote the rate of the Poisson process for task generation at UE $m$ in time slot $t$. In each time slot, $\pi_m^t$ is randomly drawn from $\pi^t_m\in[0,\pi_{\max}]$ to capture the temporal variation in task arrival pattern. Let $\bm{\pi}^t=\{\pi^t_m\}_{m\in\mathcal{M}}$ denote the task arrival pattern of all UEs in time slot $t$. The UEs may request for different types of tasks which vary in input date size and required CPU cycles. To simply the system model, we assume that the expected input data size for one task is $s$ (in bits) and the expected number of CPU cycles required by one task is $h$. The total task arrival rate to SBS $i$, denoted by $\phi^t_i$, is $\phi^t_i = \sum _{m\in\mathcal{M}_i}\pi^t_m$. The task arrival rates to all SBSs are collected in $\bm{\phi}^t=\{\phi^t_i\}_{i\in\mathcal{N}}$. 

\subsection{Transmission model}
\subsubsection{Transmission energy consumption} Transmissions occur on both the wireless link between UEs and SBSs, and the wired link among SBSs. Usually the energy consumption of wireless transmission dominates and hence we consider only the wireless part. In each time slot $t$, SBSs have to serve both uplink and downlink traffic data. We assume that the uplink and downlink transmission operate on orthogonal channels and focus on the downlink traffic since energy consumption of SBSs is mainly due to downlink transmission. Suppose each SBS $i \in\mathcal{N}$ operates at a fixed transmission power $P^d_{i}$, then the achievable downlink transmission rate $r^{d,t}_{im}$ between UE $m$ and SBS $i$ is given by the Shannon capacity, 

\begin{equation}\label{eq:shannon}
r^{d,t}_{im}=W\log_2\left(1+\frac{P^d_iH^t_{im}}{\sigma^2}\right),
\end{equation}
where $W$ is the channel bandwidth, $H^t_{im}$ is the channel gain between SBS $i$ and UE $m$, and $\sigma^2$ is the noise power. The downlink traffic consists of the computation result and other communication traffic. Since the size of computation result is usually very  small, we only consider the communication traffic. Let the downlink traffic size in time slot $t$ be $w^t_{m}\in[0,w_{\max}]$, then the energy consumption of SBS $i$ for wireless transmission is
\begin{equation}
E^{\text{tx},t}_i=\sum\limits_{m\in\mathcal{M}_i}\dfrac{ P^d_{i}w^t_{m}}{r^{d,t}_{im}}.
\end{equation}

\subsubsection{UE-to-SBS transmission delay}
The transmission delay is incurred during UE-to-SBS offloading where UEs send computation tasks to the SBSs through the uplink channel. Let $P^{u}_m$ be the transmission power of UE $m$, then the uplink transmission rate between UE $m$ and SBS $i$, denoted by $r^{u,t}_{im}$, can be obtained similarly as in $\eqref{eq:shannon}$. Therefore, the total transmission delay cost for UEs covered by SBS $i$ can be calculated as
\begin{align}
	D^{u,t}_i=\sum_{m\in\mathcal{M}_i} \frac{s \pi^t_m}{r^{u,t}_{im}}.
\end{align}

\subsection{SBS Peer Offloading}
Since workload arrivals are often uneven among the SBSs, computation offloading between peer SBSs can be enabled to exploit underused, otherwise wasted, computational resource to improve the overall system efficiency. We assume that tasks can be offloaded only once: if a task is offloaded from SBS $i$ to SBS $j$, then it will be processed at SBS $j$ and will not be offloaded further or back to SBS $i$ to avoid offloading loops. Let $\bm{\beta}^t_{i\cdot}=\{\beta^t_{ij}\}_{j\in\mathcal{N}}$ denote the offloading decision of SBS $i$ in time slot $t$, where $\beta^t_{ij}$ denotes the fraction of received tasks offloaded from SBS $i$ to SBS $j$ (notice that $\beta^t_{ii}$ is the fraction that SBS $i$ retains). A peer offloading profile of the whole system is therefore $\bm{\beta}^t=\{\bm{\beta}^t_{i\cdot}\}_{i\in\mathcal{N}}$. We further define $\bm{\beta}^t_{\cdot i}=\{\beta^t_{ji}\}_{j\in\mathcal{N}}$ as the inbound tasks of SBS $i$, namely the tasks offloaded to SBS $i$ from other SBSs. Clearly, the total workload that will be processed by SBS $i$ is $\omega^t_i(\bm{\beta}^t) \triangleq \sum_{j\in\mathcal{N}}\beta^t_{ji} $. To better differentiate the two types of workload $\phi^t_i$ and $\omega^t_i(\bm{\beta}^t)$, we call $\phi^t_i$ the \textit{pre-offloading} workload and $\omega^t_i(\bm{\beta}^t)$ the \textit{post-offloading} workload.  A profile $\bm{\beta}^t$ is feasible if it satisfies:
\begin{enumerate}
  \item \emph{Positivity}: $\beta^t_{ij}\geq 0$, $\forall i,j\in\mathcal{N}$. The offloaded workload must be non-negative.
  \item \emph{Conservation}: $\sum_{j=1}^{N}\beta^{t}_{ij}=\phi^t_{i}$, $\forall i\in\mathcal{N}$. The total offloaded workload (including the retained workload) by each SBS must equal its \textit{pre-offloading} workload.
  \item  \emph{Stability}: $\omega^t_i(\bm{\beta}^t) \leq f_{i}/h$, $\forall i\in\mathcal{N}$. The \textit{post-offloading} workload of each SBS must not exceed its service rate.
\end{enumerate}
Let $\mathcal{B}^t$ denote the set of all feasible peer offloading profile.

Since the bandwidth of the LAN is limited, peer offloading also causes additional delay due to network congestion. We assume that the expected congestion delay depends on the total traffic through the LAN, denoted by $\lambda^t(\bm{\beta}^t)=\sum_{i\in\mathcal{N}}\lambda^t_i(\bm{\beta}^t)$, where $\lambda^t_i(\bm{\beta}^t)=\sum_{j\in\mathcal{N}\backslash\{i\}}\beta_{ij}=\phi^t_i-\beta^t_{ii}$ is the number of tasks offloaded to other SBSs from SBS $i$. We assume the data size of computation tasks has a exponential distribution, then the congestion delay $D^{g,t}$ is modeled as a M/M/1 queuing system \cite{cooper1981introduction}:
\begin{align}
D^{g,t}(\bm{\beta}^t)=\dfrac{\tau}{1-\tau\lambda^t(\bm{\beta}^t)},\quad\lambda^t<\dfrac{1}{\tau},
\end{align}
where $\tau$ is the expected delay for sending and receiving $s$ bits (i.e., expected input data size of a computation task) over the LAN without congestion.

\subsection{Computation model}

\subsubsection{Computation delay}
The computation delay is due to the limited computing capability of SBSs. UEs in the network may request different types of services, therefore the required number of CPU cycles to process a computation task may vary across tasks. We model the distribution of the required number of CPU cycles of individual tasks as an exponential distribution. Given the constant processing rate, the service time of a task therefore follows an exponential distribution. Further considering the Poisson arrival of the computation tasks, the computation delay at each SBS can be modeled as an M/M/1 queuing system \cite{cooper1981introduction} and the expected computation delay $D^{f,t}_{i}$ for one task at SBS $i$ is:
\begin{align}
D^{f,t}_i(\bm{\beta}^t)=\dfrac{1}{\mu_i-\omega^t_i(\bm{\beta}^t)},
\end{align}
where $\mu_i=f_i/h$ is the expected service rate with regard to the number of tasks (i.e. tasks per second) and $\omega^t_i(\bm{\beta}^t)$ is the workload processed at SBS $i$ given the peer offloading decision $\bm{\beta}^t$. Figure \ref{fig:queuedelay} illustrates the relation between the computation delay and the congestion delay.

\begin{figure}[htb]
	\centering
	\includegraphics[width=0.5\linewidth]{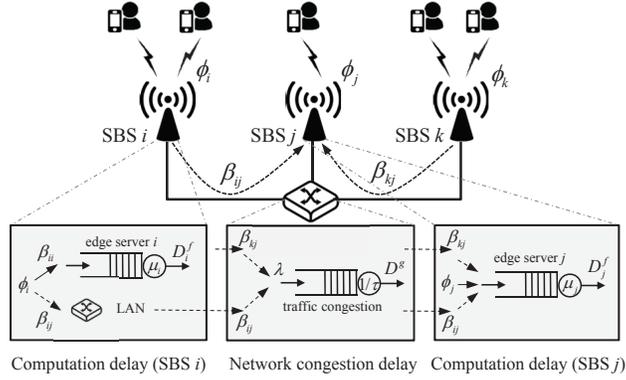}
	\vspace{-0.1 in}
	\caption{Illustration of system delay with queuing models. The figure gives a example of network with three SBSs, where SBS $i$ and SBS $k$ offload workload $\beta_{ij}$ and $\beta_{kj}$ to SBS $j$, respectively. The computation delay model for SBS $k$ is omitted since it is the same as that for SBS $i$.}
	\label{fig:queuedelay}
\end{figure}

\subsubsection{Computation energy consumption}
The computation energy consumption at SBS $i$ is load-dependent, denoted as $E^{c,t}_i$. In this paper, we consider a linear computation energy consumption function $E^{c,t}_{i}(\bm{\beta}^t)=\kappa\cdot \omega^t_i(\bm{\beta}^t)$, where $\kappa>0$ is the energy consumption for executing $h$ (i.e., expected number of CPU cycles required by a computation task) CPU cycles.

\subsection{Problem Formulation}
Peer offloading relies on SBSs' cooperative behavior in sharing their computing resources as well as their energy costs. A large body of literature was dedicated to design incentive mechanisms \cite{liu2011efficient,Mihailescu2010oneconomic} to encourage cooperation among self-interested SBSs (e.g. computing capability, energy budget) to improve the social welfare. The focus of our paper is not to design yet another incentive mechanism. Instead, we design SBS peer offloading strategies taking the SBS committed resources as the input and hence our method can work in conjunction with any existing incentive/cooperation mechanisms. Usually, the decision cycle of the resource scheduling is much longer than that of peer offloading, therefore in this paper we consider each SBS has a predetermined long-term energy consumption constraint as a result of some incentive mechanism.

Given the system model, the total delay cost of SBS $i$, defined as a sum of the delays experienced by the tasks arrived at SBSs $i$, consists of computation delay cost, network congestion delay cost, and UE-to-SBS transmission delay cost:
\begin{equation}
\begin{split}
& D^t_i(\bm{\beta}^t)=\sum_{j\in\mathcal{N}}\beta^t_{ij}D^{f,t}_{j}(\bm{\beta}^t)+\lambda^t_i D^{g,t}(\bm{\beta}^t)+D^{u,t}_{i}\\
&=\sum_{j\in\mathcal{N}}\dfrac{\beta^t_{ij}}{\mu_j-\sum_{k\in\mathcal{N}}\beta^t_{kj}}+\dfrac{\tau\lambda^t_i(\bm{\beta}^t)}{1-\tau\sum_{k\in\mathcal{N}}\lambda^t_k(\bm{\beta}^t)}+D^{u,t}_{i},
\end{split}
\end{equation}
and the energy consumption of SBS $i$ in time slot $t$ consists of transmission energy consumption and computation energy consumption:
\begin{equation}
	E^t_i(\bm{\beta}^t)=E^{\text{tx},t}_i+E^{c,t}_{i}(\bm{\beta}^t)=E^{\text{tx},t}_i+\kappa\sum_{k\in\mathcal{N}}\beta^t_{ki} .
\end{equation}
The objective of network operator is to minimize the long-term system delay cost given the energy budgets committed by individual SBSs (which are outcomes of the adopted incentive mechanisms). Formally, the problem is
\begin{subequations}
	\begin{align}
	\textbf{P1} \quad & \min_{\bm{\beta}^1,\dots,\bm{\beta}^{T-1}}\frac{1}{T}\sum\limits_{t=0}^{T-1}\sum_{i=1}^{N}\mathbb{E}\left\{D_i^t(\bm{\beta}^t)\right\}\\
	\text{s.t.}\quad
	&\frac{1}{T}\sum\limits_{t=0}^{T-1}\mathbb{E}\left\{E^t_i(\bm{\beta}^t)\right\}\leq \bar{E}_i, \forall i\in \mathcal{N} \label{longtermE}\\
	& E^t_i(\bm{\beta}^t)\leq E_{\max}, \forall i\in\mathcal{N}, \forall t \label{emax}\\
	& D_i^t(\bm{\beta}^t)\leq D_{\max}, \forall i\in\mathcal{N}, \forall t \label{dmax}\\
	& \bm{\beta}^t \in \mathcal{B}^{t},\forall t \label{feasible}
	\end{align}
\end{subequations}
Constraint \eqref{longtermE} is the long-term energy budget constraint for each SBS.  Constraint \eqref{emax} requires that the energy consumption of a SBS does not exceed an upper limit $E_{\max}$ in each time slot. Constraint \eqref{dmax} indicates that the per-slot delay of each SBS is capped by an upper limit $D_{\max}$ so that the real-time performance is guaranteed in the worst case.

The major challenge that impedes the derivation of optimal solution to \textbf{P1} is the lack of future information. Optimally solving \textbf{P1} requires complete offline information (task arrivals across all time slots) which is difficult to predict in advance, if not impossible. Moreover, the long-term energy constraints couple the peer offloading decision across different slots: consuming more energy in the current slot will reduce the available energy for future use. These challenges call for an online optimization approach that can efficiently perform peer offloading without foreseeing the future.

\section{Online SBS Peer Offloading}\label{sec_online}
In this section, we develop a novel framework for making online SBS peer offloading decisions, called OPEN (Online SBS PEer offloadiNg) by leveraging the Lyapunov technique. OPEN converts \textbf{P1} to per-slot optimization problems solvable with only current information. We consider both the case in which the network operator coordinates the SBS peer offloading in a centralized way (this section) and the case in which SBSs make peer offloading decisions among themselves in an autonomous manner (next section).
\subsection{Lyapunov optimization based online algorithm}
In the optimization problem \textbf{P1}, the long-term energy constraints of SBSs couple the peer offloading decisions across times slots. To address this challenge, we leverage the \emph{Lyapunov drift-plus-penalty technique} \cite{neely2010stochastic} and construct a (virtual) energy deficit queue for each SBS to guide the peer offloading decisions to follow the long-term energy constraints. We define a set of energy deficit queues $\bm{q}(t)=\{q_i(t)\}_{i\in\mathcal{N}}$, one for each SBS, and let $q_i(0)=0, \forall i\in\mathcal{N}$. For each SBS $i\in\mathcal{N}$, its energy deficit queue evolves as follows:
\begin{equation}\label{queue}
	q_i(t+1)=\max\{q_i(t)+E^t_i(\bm{\beta}^t)-\bar{E}_i,0\},
\end{equation}
where $q_i(t)$ is the queue length in time slot $t$, indicating the deviation of current energy consumption from the long-term energy constraint of SBS $i$.

\begin{algorithm}[htb]\label{alg_OPEN}
	\caption{OPEN}
	\KwIn{control parameter $V$, energy deficit queues $\bm{q}(0)=\bm{0}$;}
	\KwOut{offloading decisions $\bm{\beta}^0,\dots,\bm{\beta}^{T-1}$;}
	\For{$t=0$ \textbf{to} $T-1$ }
	{   Observe workload arrival $\bm{\phi}^t$ and feasible peer offloading strategy set $\mathcal{B}^t$ \;
		Solving \textbf{P2} to get optimal $\bm{\beta}^t$ in time slot $t$: ~~~~~~~~
		$\min\limits_{\bm{\beta}^t\in\mathcal{B}^{t}}~\sum\limits_{i\in\mathcal{N}} \left(V\cdot D^t_i(\bm{\beta}^t)+q_i(t) \cdot E^t_i(\bm{\beta}^t)\right)$ \label{Line:obj}\;
		Update the deficit for all SBS $i$:\\
		$q_i(t+1)=[q_i(t)+ E^t_i(\bm{\beta}^t)-\bar{E}_i]^+$
	}
	\textbf{return} $\bm{\beta}^1,\dots,\bm{\beta}^{T-1}$\;
\end{algorithm}

Next, we present the online algorithm OPEN (Algorithm \ref{alg_OPEN}) for solving \textbf{P1}. In OPEN, the network operator determines the peer offloading strategy in each time slot $t$ by solving the optimization problem \textbf{P2}, as presented below:
\begin{align*}
\textbf{P2}~\min_{\bm{\beta}^t\in\mathcal{B}^t}~\sum_{i=1}^{N}&\left(V\cdot D^t_i(\bm{\beta}^t)+q_i(t) \cdot E^t_i(\bm{\beta}^t)\right)\\
\text{s.t.}&~\eqref{emax}, \eqref{dmax} ~\text{and}~\eqref{feasible}
\end{align*}

The objective in \textbf{P2} is designed based on \textit{Lyapunov drift-plus-penalty} framework. The rationale behind this design will be explained later in Section \ref{subsec:performance_analysis}. The first term in \textbf{P2} is to minimize the system delay and the second term is added aiming to satisfy the long-term energy constraint \eqref{longtermE} in an online manner; the positive control parameter $V$ is used to adjust the trade-off between these two purposes. To give a brief explanation, by considering the additional term $\sum_{i=1}^{N}q_i(t) E^t_i(\bm{\beta}^t)$, the network operator takes into account the energy deficits of SBSs in current-slot decision making: when $\bm{q}(t)$ is larger, minimizing the energy deficits is more critical for network operator. Thus, OPEN works following the philosophy of ``if violate the energy budget, then use less energy'', and hence the long-term energy constraint can be satisfied in the long run without foreseeing the future information. Later in this section, we will rigorously prove the performance of OPEN in terms of system delay cost and long-term energy consumption. Now, to complete OPEN, it remains to solve the optimization problem \textbf{P2}. Notice that solving \textbf{P2} requires only currently available information as input.

\subsection{Centralized solution to OPEN}
In this subsection, we consider the existence of a centralized controller who collects the complete current-slot information from all SBSs, solves the per-slot problem \textbf{P2}, and coordinates SBS peer offloading in each time slot $t$. Before proceeding to the solution, we rewrite the objective function of \textbf{P2} as below:
\begin{align}\label{soical_opt_obj}
&\sum\limits_{i\in\mathcal{N}}\left(V \cdot D^t_i(\bm{\beta}^t)+q_i(t) \cdot E^t_i(\bm{\beta}^t)\right) \nonumber\\
=&\sum_{i\in\mathcal{N}}V\cdot(\sum_{j\in\mathcal{N}}\beta^t_{ij}D^{f,t}_{j}(\bm{\beta}^t)+\lambda^t_i (\bm{\beta}^t)D^{g,t}(\bm{\beta}^t)+D^{u,t}_{i})+\sum_{i\in\mathcal{N}}q_i(t)(E^{\text{tx},t}_i+E^{c,t}_i(\bm{\beta}^t)) \nonumber\\
=&\sum_{j\in\mathcal{N}}\dfrac{V\sum_{i=1}^N\beta^t_{ij}}{\mu_j-\sum_{k=1}^{N}\beta^t_{kj}}+\dfrac{V\tau\sum_{i=1}^{N}\lambda^t_i(\bm{\beta}^t)}{1-\tau\sum_{k=1}^{N}\lambda^t_k(\bm{\beta}^t)}  +\sum_{i\in\mathcal{N}}\kappa q_i(t)\omega^t_i(\bm{\beta}^t)+\sum_{i\in\mathcal{N}}\left(VD^{u,t}_{i}+q_i(t)E^{\text{tx},t}_i\right) \nonumber \\
=&\underbrace{\sum_{i\in\mathcal{N}}V\left(\dfrac{V\omega^t_i({\bm{\beta}^t})}{\mu_i-\omega^t_i(\bm{\beta}^t)}+\kappa q_i(t)\omega^t_i(\bm{\beta}^t) \right)+\dfrac{V\tau\lambda^t({\bm{\beta}^t})}{1-\tau\lambda^t({\bm{\beta}^t})}}_{\text{decision-dependent}}+\underbrace{\sum_{i\in\mathcal{N}}\left(VD^{u,t}_i+q_iE^{\text{tx},t}_i\right)}_{\text{decision-independent}}.
\end{align}
The objective function of \textbf{P2} can be divided into two parts: ($i$) a decision-dependent part which is a weighted sum of the computation delay cost, the computation energy consumption, and the network congestion delay cost; ($ii$) a decision-independent part which relates to the UE-to-SBS transmission delay and SBS-to-UE energy consumption. Therefore, we focus on the decision-dependent part for solving \textbf{P2}. Although the decision-independent part does not affect the solution of \textbf{P2} directly, its second term (i.e., $E^{\text{tx},t}$) will affect the energy deficit queue updating and hence indirectly affects peer offloading decisions in the long-run. 

Notice that \textbf{P2} is solved in each time slot, for ease of exposition, we drop the time index for variables. Moreover, instead of optimizing $\bm{\beta}^t$ directly, we alternatively optimize the amount of workload $\omega^t_i(\bm{\beta})$ each SBS should accommodate, and the corresponding total traffic in the LAN $\lambda^t(\bm{\beta})$. By rewriting $\omega^t_i(\bm{\beta})$ as $\omega_i$  and $\lambda^t(\bm{\beta})$ as $\lambda$, \textbf{P2} is therefore equivalent to:
\begin{subequations}
	\begin{align} \textbf{P2-S}~~
	\min\limits_{\omega_i\in\Omega_i, \lambda\in\Lambda} & \sum_{i\in\mathcal{N}}\left(\dfrac{V\omega_i}{\mu_i-\omega_i}+\kappa q_i\omega_i\right)+\dfrac{V\tau\lambda}{1-\tau\lambda}\\
	\text{s.t.} ~~ & E_i(\omega_i,\lambda)\leq E_{\max}, \forall i\in\mathcal{N} \label{eq:rew_emax}\\
	& D_i(\omega_i,\lambda)\leq D_{\max}, \forall i\in\mathcal{N} \label{eq: rew_dmax}\\
	& \omega_i\in\Omega_i, \lambda\in\Lambda,\forall i\in\mathcal{N} \label{eq:rew_feasible}
	\end{align}
\end{subequations}
where $\Omega_i$ in \eqref{eq:rew_feasible} is the feasible space for $\omega_i$ determined by the mapping $\omega_i: \mathcal{B}\to\Omega_i$, and similarly $\Lambda$ is determined by $\lambda: \mathcal{B}\to\Lambda$. Notice that, although $\omega_i$s and $\lambda$ are written as independent variables, they are deterministic functions of a particular $\bm{\beta}$ in each time slot. To capture the relation between $\omega_i$s and $\lambda$, we introduce and closely follow a workload flow equation when solving \textbf{P2-S}, which will be shown shortly.

Next, we give the optimal solution for the above optimization problem starting with classifying SBSs into the following three categories:
\begin{itemize}
	\item \textbf{Source SBS} ($\mathcal{R}$). A SBS is a source SBS if it offloads a positive portion of its \textit{pre-offloading} workloads to other SBSs and processes the rest of workloads locally. Moreover, it does not receive any workload from other SBSs ($0\leq\omega_i<\phi_i$).
	\item \textbf{Neutral SBS} ($\mathcal{U}$): A SBS is a neutral SBS if it processes all its \textit{pre-offloading} workloads locally and does not receive any workload from other SBSs ($\omega_i=\phi_i$).
	\item \textbf{Sink SBS} ($\mathcal{S}$): A SBS is a sink SBS if it receives workloads from other SBSs and does not offload workload to others ($\omega_i>\phi_i$).
	\end{itemize}

Notice that in our categorization, there is no SBS such that it offloads workloads to other SBSs while receiving workloads from other SBSs. This is because it can be easily shown that having such SBSs result in suboptimal solutions to $\textbf{P2}$ due to the extra network congestion delay. To assist the presentation of the optimal solution, we define two auxiliary functions.
\begin{definition}
Define $d_i(\omega_i) \triangleq \frac{\partial}{\partial \omega_i}[\omega_iD^f_i(\omega_i)]=\frac{\mu_i}{(\mu_i-\omega_i)^2}$ as the marginal computation delay function for SBS $i, \forall i\in\mathcal{N}$; $g(\lambda)\triangleq\frac{\partial}{\partial\lambda}[\lambda D^{g}(\lambda)]=\frac{\tau}{(1-\tau\lambda)^2}$ as the marginal congestion delay function.
\end{definition}
Specifically, $d_i(\omega_i)$ is the marginal value of the computation delay function when $\omega_i$ tasks are processed at SBS $i$; and $g(\lambda)$ is the marginal value of the congestion delay function with $s\cdot\lambda$ bits traffic in the LAN.

We define $\xi_i \triangleq Vd_i(\phi_i) + \kappa q_i$ as the \textit{pre-offloading} Marginal Computation Cost (MaCC), taking into account both the computation delay cost and the computation energy consumption if SBS $i$ processes all its tasks locally. Based on $\xi_i$, Theorem \ref{centralized_solution} shows the optimal SBS categorization, workload allocation, and corresponding traffic in LAN.

\begin{theorem} \label{centralized_solution} The category that SBS $i$ belongs to, the optimal post-offloading workload $\omega^*_i$, and the corresponding traffic in LAN $\lambda^*$ can be determined based on pre-offloading MaCC $\xi_i$ and a parameter $\alpha$: \\
	(a) If $\xi_i<\alpha$, then $i\in \mathcal{S}$ and $\omega^*_i=d^{-1}_i(\frac{1}{V}(\alpha-\kappa q_i))$;\\
	(b) If $\alpha\leq \xi_i \leq \alpha+Vg(\lambda^*)$, then $i\in \mathcal{U}$ and $\omega_i^*=\phi_i$;\\
	(c) If $\xi_i>\alpha+Vg(\lambda^*)$, then $i\in \mathcal{R}$ and $\omega^*_i=[d^{-1}_i(\frac{1}{V}(\alpha+ \allowbreak ~~~~~ Vg(\lambda^*)-\kappa q_i))]^+$;\\
where $\lambda^*$, $\alpha$ are the solution to the workload flow equation
\begin{align}
\underbrace{\sum_{i\in\mathcal{S}}\left(d^{-1}_i(\frac{1}{V}(\alpha-\kappa q_i))-\phi_i\right)}_{\lambda^S \emph{: inbound workloads to sinks}}=\underbrace{\sum_{i\in\mathcal{R}}\left(\phi_i-[d^{-1}_i(\frac{1}{V}(\alpha+Vg(\lambda^*)-\kappa q_i))]^+\right)}_{\lambda^R \emph{: outbound workloads from sources}}.
\end{align}
\end{theorem}
\begin{proof}
See Appendix A in Supplementary File.
\end{proof}
In Theorem \ref{centralized_solution}, $\alpha$ is a Lagrange multiplier that equals the unique optimal \textit{post-offloading} MaCC (i.e. $Vd_i(\omega^*_i) + \kappa q_i$) of sink SBSs. Part ($a$) indicates that SBSs with \textit{pre-offloading} MaCC less than $\alpha$  will serve as sink SBSs and their post-offloading MaCCs will be equal to $\alpha$. Part ($b$) implies that the \textit{pre-offloading} MaCC of a neutral SBS is no less than $\alpha$ but no larger than the sum of $\alpha$ and the marginal congestion delay cost. This means that no other SBSs would benefit from offloading workloads to neutral SBSs and at the same time neutral SBS receive no benefits by performing peer offloading. For a source SBS, its \textit{pre-offloading} MaCC is larger than the sum of $\alpha$ and the marginal congestion delay cost and therefore, it tends to offload workloads to other SBSs until its \textit{post-offloading} MaCC reduces to $\alpha+Vg(\lambda^*)$ (i.e. $\omega^*_i=d^{-1}_i(V^{-1}(\alpha+Vg(\lambda^*)-\kappa q_i))$ or no more workload can be further offloaded (i.e. $\omega^*_i=0$). Figure \ref{centralized} depicts the categorization of SBSs and the difference between their \textit{pre-offloading} and \textit{post-offloading} MaCCs.
\begin{figure}[htb]
	\centering	
	\includegraphics [width=0.75\linewidth]{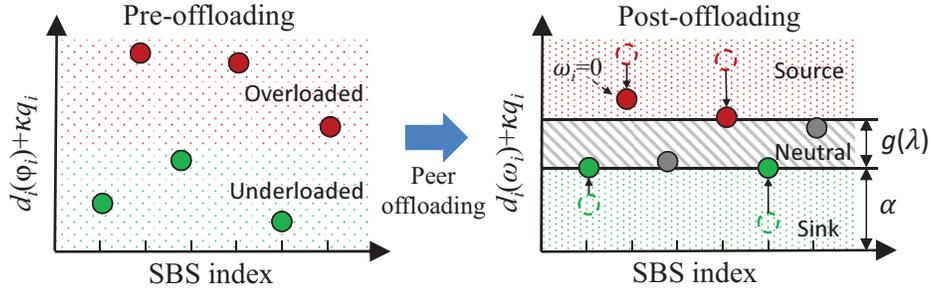}
	\vspace{-0.1in}
	\caption{SBS categorization and changes in the marginal computation costs. }
	\label{centralized}
\end{figure}

\begin{algorithm}\label{COPEN}
	\caption{OPEN-Centralized}
	\KwIn{SBS services rates: $\mu_1, \mu_2, \dots, \mu_N$; SBS workload arrival rates: $\phi_1,\phi_2,\dots,\phi_N$; Network mean communication time: $\tau$.}
	Initialization: $\omega_i\leftarrow\phi_i, i\in\mathcal{N}$\;
	Get pre-offloading MaCC for SBSs $\xi_i=Vd_i(\phi_i)+\kappa q_i$\; 
	Find $\xi_{\max}=\max_{i\in\mathcal{N}} \xi_i$ and $\xi_{\min}=\min_{i\in\mathcal{N}}\xi_i$\;
	\textbf{If} $\xi_{\min}+Vg(0) \geq\xi_{\max}$, 
	then \textbf{STOP} (no peer offloading is required)\;
	$a\leftarrow \xi_{\min}$; $b\leftarrow \xi_{\max}$\;
	\While {$|\lambda^S(\alpha)-\lambda^R(\alpha)|\geq \tilde{\epsilon}$}
	{ $\lambda^S(\alpha)\leftarrow 0$,$~\lambda^R(\alpha)\leftarrow 0$\;
		$\alpha\leftarrow \dfrac{1}{2}(a+b)$\;
		Get in order: $\mathcal{S}(\alpha)$, $\lambda^S(\alpha)$, $\mathcal{R}(\alpha)$, $\mathcal{U}(\alpha)$, $\lambda^R(\alpha)$ \;
		\textbf{if} $\lambda^S(\alpha)>\lambda^R(\alpha)$ \textbf{then} $b\leftarrow\alpha$\;
		\textbf{else} $a\leftarrow\alpha$\;
	}
	Determine $\omega_i^*$ and $\lambda^*$ based on Theorem \ref{centralized_solution}\;
\end{algorithm}

Since directly deriving a solution for $\alpha$ is impossible, we develop an iterative algorithm called OPEN-Centralized for solving \textbf{P2}. OPEN-Centralized uses binary search to obtain $\alpha$ under the workload flow equation, which is summarized in Algorithm 2. In each iteration, the algorithm first determines a set of sink SBSs ($\mathcal{S}$) according to the parameter $\alpha$, then the corresponding amount of inbound workload is determined. Given the total workload $\lambda=\lambda^S$ transmitted in the LAN, the algorithm determines the source SBSs ($\mathcal{R}$), neutral SBSs ($\mathcal{U}$) and then calculates the outbound workload $\lambda^R$. If $\lambda^R$ equals $\lambda^S$, then the optimal $\alpha$ is found; otherwise, the algorithm updates $\alpha$ and goes into the next iteration. 

\textbf{Remark:} Notice that the algorithm outputs the optimal workload allocation of SBSs, $\omega^*_i, \forall i\in\mathcal{N}$ and the corresponding traffic in the LAN $\lambda^*$. Any peer offloading strategy that realizes the optimal workload allocation is an optimal peer offloading profile $\bm{\beta}^*$ for \textbf{P2}.

\subsection{Performance Analysis of OPEN}\label{subsec:performance_analysis}
In this section, we rigorously prove the performance bound of OPEN. We first show that with the designed energy deficit queue $q_i(t)$, the long-term energy constraint \eqref{longtermE} for each SBS is enforced if the queues $q_i(t)$ are stable, i.e. $\lim_{T\rightarrow\infty}\mathbb{E}\{q_i(t)\}/T=0$. 

According to the queue update rule in \eqref{queue}, we have $q_i(t+1)\geq q_i(t)+E_i^t-\bar{E}_i$. Summing this over time slots $t\in\{0,1,\dots,T-1\}$ then dividing the sum by $T$, we have
\begin{equation}
\dfrac{q_i(T)-q_i(0)}{T}+\bar{E}_i\geq\dfrac{1}{T}\sum_{t=0}^{T-1}E^t_i(\bm{\beta}^t).
\end{equation}
Initializing the queue as $q_i(0)=0, \forall i$ and letting $T\rightarrow \infty$, then taking the expectations of both side yields
\begin{equation}\label{stable_ineq}
\lim_{T\rightarrow\infty}\dfrac{\mathbb{E}\{q_i(t)\}}{T}+\bar{E}_i\geq\lim_{T\rightarrow\infty}\dfrac{1}{T}\sum_{t=0}^{T-1}\mathbb{E}\{E^t_i(\bm{\beta}^t)\}.
\end{equation}	
If the virtual queues $q_i(t)$ are stable (i.e., $\lim_{T\rightarrow\infty} \mathbb{E} \{q_i(t)\}/T \allowbreak =0$), then \eqref{stable_ineq} becomes the long-term energy constraint \eqref{longtermE}. In the following, it will be shown that the deficit queue $q_i(t)$ is guaranteed to be stable by running the proposed OPEN. We start with defining a \emph{quadratic Lyapunov function} $L(\bm{q}(t))$ as: $
L(\bm{q}(t))\triangleq\frac{1}{2}\sum_{i=1}^{N}q_i^2(t)$. It represents a scalar metric of the queue length in all virtual queues. A small value of $L(\bm{q}(t))$ implies that all the queue backlogs are small, which means the virtual queues have strong stability. To keep the virtual queues stable (i.e., to enforce the energy constraints) by persistently pushing the Lyapunov function towards a lower value, we introduce \emph{one-slot Lyapunov drift} $\Delta(\bm{q}(t))$:
\begin{equation*}
\begin{split}
&\Delta(\bm{q}(t))\triangleq \mathbb{E}\left\{L(\bm{q}(t+1))-L(\bm{q}(t))|\bm{q}(t)\right\}\\
&=\dfrac{1}{2}\sum_{i=1}^{N}\mathbb{E}\left\{q_i^2(t+1)-q_i^2(t)|\bm{q}(t)\right\}\\
&\stackrel{(\dag)}{\leq}\dfrac{1}{2}\sum_{i=1}^{N}\mathbb{E}\left\{(q_i(t)+E^t_i(\bm{\beta}^t)-\bar{E}_i)^2-q_i^2(t)|\bm{q}(t)\right\}\\
&\leq\dfrac{1}{2}\sum_{i=1}^{N}(E_{\max}-\bar{E}_i)^2+\sum_{i=1}^{N}q_i(t)\mathbb{E}\left\{E^t_i(\bm{\beta}^t)-\bar{E}_i|\bm{q}(t)\right\}.
\end{split}
\end{equation*}
The inequality ($\dag$) comes from $(q_i(t)+E^t_i(\bm{\beta}^t)-\bar{E}_i)^2 \geq [\max(q_i(t)+E^t_i(\bm{\beta}^t)-\bar{E}_i,0)]^2$. By extending the Lyapunov drift to an optimization problem, our objective is to minimize a supremum bound on the following \emph{drift-plus-penalty} expression in each time slot:
\begin{align}\label{drift_plus_cost}
\Delta(\bm{q}(t))+ & V\sum_{i=1}^{N}\mathbb{E}\left\{D_i^t(\bm{\beta}^t)|\bm{q}(t) \right\} \leq V\sum_{i=1}^{N}\mathbb{E}\left\{D_i^t(\bm{\beta}^t)|\bm{q}(t)\right\} \nonumber \\
& + B +  \sum_{i=1}^{N}q_i(t)\mathbb{E}\left\{(E^t_i(\bm{\beta}^t)-\bar{E}_i)|\bm{q}(t) \right\},
\end{align}
where $B=\frac{1}{2}\sum_{i=1}^{N}(E_{\max}-\bar{E}_i)^2$. Notice that OPEN (Line \ref{Line:obj} in Algorithm \ref{alg_OPEN}) exactly minimizes the right hand side of \eqref{drift_plus_cost}. The parameter $V\geq 0$ controls the \emph{delay-energy deficit} tradeoff, i.e., how much we shall emphasize the delay minimization compared to the energy deficit. Next we give a rigorous performance bound of OPEN compared to the optimal solution to $\textbf{P1}$.

\begin{theorem}\label{Theorem_Online}
Following the optimal peer offloading decision $\bm{\beta}^{*,t}$ obtained by OPEN-Centralized, the long-term system delay cost satisfies:
\begin{equation}\label{system_delay_bound}
 \lim_{T\rightarrow\infty}\dfrac{1}{T}\sum_{t=0}^{T-1}\sum_{i=1}^{N}\mathbb{E}\left\{D_i^t(\bm{\beta}^{*,t})\right\}<D^{\text{opt}}_{\text{sys}}+\dfrac{B}{V},
\end{equation}
and the long-term energy deficit of SBSs satisfies:
\begin{align}\label{energy_defit_bound}
\lim_{T\rightarrow\infty}\dfrac{1}{T}\sum_{t=0}^{T-1}\sum_{i=1}^{N} & \mathbb{E}\left\{E^t_i(\bm{\beta}^{*,t})-\bar{E}_i\right\}\leq\dfrac{1}{\epsilon}\left(B+V(D^{\max}_{\text{sys}}-D^{\text{opt}}_{\text{sys}})\right),
\end{align}
where $D^{\text{opt}}_{\text{sys}}=\lim_{T\to\infty}\frac{1}{T}\sum_{t=0}^{T-1} \sum_{i=1}^{N}\mathbb{E}\left\{D_i^t(\bm{\beta}^{\text{opt},t})\right\}$ is the optimal system delay to $\textbf{\emph{P1}}$, $D^{\max}_{\text{sys}}=ND_{\max}$ is the largest system delay cost, and $\epsilon>0$ is a constant which represents the long-term energy surplus achieved by some stationary strategy.
\end{theorem}

\begin{proof}
See Appendix B in Supplementary File.
\end{proof}
The above theorem demonstrates an $[O(1/V),O(V)]$ \emph{delay-energy deficit tradeoff}. OPEN asymptotically achieves the optimal performance of the offline problem $\textbf{P1}$ by letting $V \rightarrow \infty$. However, the optimal system delay cost is achieved at the price of a larger energy deficit, as a larger deficit queue is required to stabilize the system and hence postpones the convergence. The long-term energy deficit bound in \eqref{energy_defit_bound} implies that the time-average energy deficit grows linearly with $V$. Notice that the total energy consumption of the overall system stays almost the same regardless of the peer offloading decision. This is because all computation tasks are accommodated within the edge system and the same amount of energy will be spent to process these tasks (assuming the energy due to wired transmission is negligible). The real issue is where these tasks are processed and how much energy each SBS should spend given its long-term energy constraint.

\section{Autonomous SBS Peer Offloading}\label{sec_noncoop}
In the previous section, the network operator coordinates SBS peer offloading in a centralized way. However, the small cell network is often a distributed system where there is no central authority controlling the workload allocation. Moreover, individual SBSs may not have the complete information of the system, which impedes the derivation of social optimal solution. In this section, we formulate OPEN as a non-cooperative game where SBSs minimize their own costs in a decentralized and autonomous way. We analyze the existence of Nash Equilibrium (NE) and the efficiency loss due to decentralized coordination compared to the centralized coordination in terms of the \emph{Price of Anarchy} (PoA).

\subsection{Game Formulation}
We first define a non-cooperative game $\Gamma \triangleq (\mathcal{N}, \{ \mathcal{B}_{i\cdot}\}_{i\in\mathcal{N}},$ $\{K_i\}_{i\in\mathcal{N}})$, where $\mathcal{N}$ is the set of SBSs, $\mathcal{B}_{i\cdot}$ is the set of feasible peer offloading strategies for SBS $i$, and $K_i$ is the cost function for each SBS $i$ defined as $K_i(\bm{\beta}^t)=V\cdot D^t_i(\bm{\beta}^t)+q_i(t) E^t_i(\bm{\beta}^t)$. In the autonomous scenario, each SBS aims to minimize its own cost by adjusting its own peer offloading strategy in each time slot $t$. Without causing confusions, we also drop the time index $t$ in this section. The Nash equilibrium of this game is defined as follows.

\begin{definition}[Nash equilibrium] \label{def_NE}
	A Nash equilibrium of the SBS peer offloading game defined above is a peer offloading profile $\bm{\beta}^{\text{NE}}=\{\bm{\beta}^{\text{NE}}_{i\cdot}\}_{i\in\mathcal{N}}$ such that for every SBS $i\in\mathcal{N}$:
	\begin{equation}
	\bm{\beta}^{\text{NE}}_{i\cdot}\in\arg\min_{\tilde{\bm{\beta}}_{i\cdot}}K_i(\bm{\beta}^{\text{NE}}_{1\cdot},\bm{\beta}^{\text{NE}}_{2\cdot},\dots,\tilde{\bm{\beta}}_{i\cdot},\dots,\bm{\beta}^{\text{NE}}_{N\cdot}) .
	\end{equation}	
\end{definition}
At the Nash equilibrium, a SBS cannot further decrease its cost by unilaterally choosing a different peer offloading strategy when the strategies of the other SBSs are fixed. The equilibrium peer offloading profile can be found when each SBS's strategy is a \emph{best response} to the other SBSs' strategies. Before proceeding with the analysis, we give an equivalent expression of $K_i(\bm{\beta})$ by considering only the part that depends on SBS $i$'s peer offloading strategy $\bm{\beta}_{i\cdot}$:

\begin{align}\label{eq:Ci}
& C_i(\bm{\beta}_{i\cdot}) = \underbrace{\beta_{ii}\left(\dfrac{V}{\mu_i-\omega_i(\bm{\beta}_{i\cdot},\bm{\beta}_{-i\cdot})}+\kappa q_i\right)}_{\text{cost due to local processing}} \\
& + \underbrace{ \sum_{j\in\mathcal{N} \backslash\{i\}} \beta_{ij} \left( \dfrac{V}{\mu_j - \omega_j(\bm{\beta}_{i\cdot},\bm{\beta}_{-i\cdot})}+\dfrac{V\tau}{1-\tau\lambda(\bm{\beta}_{i\cdot},\bm{\beta}_{-i\cdot})}\right)}_{\text{cost due to peer offloading}}\nonumber,
\end{align}
where $\bm{\beta}_{-i\cdot}$ is the peer offloading decisions of other SBSs except SBS $i$. The first part of \eqref{eq:Ci} is the cost incurred by processing the retained workload locally (i.e. the computation delay cost of itself and the energy consumption) and the second part is the cost incurred by performing peer offloading (i.e. the computation delay cost on other SBSs and the network congestion delay cost). To facilitate our analysis, \eqref{eq:Ci} can be further represented as:
\begin{equation}\label{eq:C_i}
C_i(\bm{\beta}_{i\cdot})=\sum_{j\in\mathcal{N}}\beta_{ij}c_{ij}(\bm{\beta}_{i\cdot}).
\end{equation}
where 
\begin{equation*}
c_{ij}(\bm{\beta}_{i\cdot})=\left\{\\
\begin{array}{ll}
	\dfrac{V}{\mu_j-\omega_j(\bm{\beta}_{i\cdot},\bm{\beta}_{-i\cdot})}+\dfrac{V\tau}{1-\tau\lambda(\bm{\beta}_{i\cdot},\bm{\beta}_{-i\cdot})}, &\text{if}~j\neq i\\
	\dfrac{V}{\mu_i-\omega_i(\bm{\beta}_{i\cdot},\bm{\beta}_{-i\cdot})}+\kappa q_i, &\text{if}~j=i\\
\end{array} \right.
\end{equation*}
Then, in the autonomous SBS peer offloading, the best response problem for each SBS $i\in\mathcal{N}$ becomes:
\begin{subequations}
	\begin{align}\label{objective_C}\textbf{P3}~~
	&\min_{\bm{\beta}_{i\cdot} \in \mathcal{B}_{i\cdot}} \sum_{j\in\mathcal{N}}\beta_{ij}c_{ij}(\bm{\beta}_{i\cdot})\\
		\text{s.t.} ~~ & E_i(\bm{\beta}_{i\cdot},\bm{\beta}_{-i\cdot})\leq E_{\max} \label{eq:P3_emax}\\
	& D_i(\bm{\beta}_{i\cdot},\bm{\beta}_{-i\cdot})\leq D_{\max} \label{eq: P3_dmax}\\
	& \bm{\beta}_{i\cdot} \in \mathcal{B}_{i\cdot} \label{eq: P3_feasible}
	\end{align}
\end{subequations}

\subsection{Existence of Nash Equilibrium}
In this section, we analyze the existence of Nash equilibrium in SBS peer offloading game. First, we define the pair-specific marginal cost functions as follows:
\begin{definition}[Pair-specific marginal cost]\label{marginalcost}
 The marginal cost function of SBS $i$ for offloading to SBS $j$ is defined as $\tilde{c}_{ij}(\bm{\beta}_{i\cdot})=\frac{\partial(\beta_{ij}c_{ij}(\bm{\beta}_{i\cdot}))}{\partial\beta_{ij}}=c_{ij}(\bm{\beta}_{i\cdot})+\beta_{ij}\frac{\partial c_{ij}(\bm{\beta}_{i\cdot})}{\partial\beta_{ij}}$.
\end{definition}
The vector of pair-specific marginal cost functions for SBS $i$ is collected in the notation $\tilde{\bm{c}}_{i\cdot}=(\tilde{c}_{ij})_{j\in\mathcal{N}}$. It is easy to verify that $\tilde{\bm{c}}_{i\cdot}(\bm{\beta}_{i\cdot})=\triangledown_i C_i(\bm{\beta}_{i\cdot})$, where $\triangledown_i C_i(\bm{\beta}_{i\cdot})$ stands for the gradient of $C_i(\bm{\beta}_{i\cdot})$ with respect to $\bm{\beta}_{i\cdot}$. The pair-specific marginal cost function for the network is collected in the notation $\tilde{\bm{c}}=(\tilde{\bm{c}}_{i\cdot})_{i\in\mathcal{N}}$. Next, we establish conditions of the existence of a Nash equilibrium via variational inequalities. To this end, we first recall a classical result which characterizes a solution of best response problem by a variational inequality.
\begin{lemma}\label{BR_var}
$\bm{\beta}^{\text{BR}}_{i\cdot}$ is the best response of SBS $i$ in \eqref{objective_C} if and only if it satisfies the following variational inequality:
	\begin{equation} \label{var_ineq}
		\langle\tilde{\bm{c}}_{i\cdot} (\bm{\beta}^{\text{BR}}_{i\cdot}), \bm{\beta}_{i\cdot}-\bm{\beta}^{\text{BR}}_{i\cdot}\rangle\geq 0, ~\forall \bm{\beta}_{i\cdot}\in\mathcal{B}_{i\cdot} ,
	\end{equation}
	where $\langle\cdot,\cdot\rangle$ is the inner product operator.
\end{lemma}
\begin{proof}
See Kinderlehrer and Stampacchia \cite{kinderlehrer2000introduction} (Proposition 5.1 and Proposition 5.2). Notice that $C_i(\bm{\beta}_{i\cdot})$ is required to be convex with respect to $\bm{\beta}_{i\cdot}$, which is obvious in \eqref{eq:C_i}.
\end{proof}

The following theorem establishes the existence of a Nash equilibrium in the SBS peer offloading game.
\begin{theorem}[Existence of Nash equilibrium]
The SBS peer offloading game admits at least one Nash equilibrium.
\end{theorem}
\begin{proof}
According to Kinderlehrer and Stampacchia \cite{kinderlehrer2000introduction} Chapter 1, Theorem 3.1, since $\mathcal{B}$ is a nonempty, compact, and convex feasible offloading strategy set and $\tilde{\bm{c}}$ as a continuous map defined on $\mathcal{B}$, there exist $\bm{\beta}^\prime\in \mathcal{B}$ that satisfies the following variational inequality:
\begin{equation}\label{sys_varInequa}
	\langle \tilde{\bm{c}}(\bm{\beta}^\prime), \bm{\beta}-\bm{\beta}^\prime\rangle\geq 0,~\forall~\bm{\beta}\in\mathcal{B}.
\end{equation}

Next, we prove that this $\bm{\beta}^\prime$ is a Nash equilibrium. According to Lemma \ref{BR_var}, for a best respond $\bm{\beta}^{\text{BR}}_{i\cdot}$ of (\ref{objective_C}), we must have:
	\begin{equation}\label{varInequa_j}
		\langle\tilde{\bm{c}}_{i\cdot} (\bm{\beta}^{\text{BR}}_{i\cdot}), \bm{\beta}_{i\cdot}-\bm{\beta}^{\text{BR}}_{i\cdot}\rangle\geq 0,
	\end{equation}
and all $\bm{\beta}^{\text{BR}}_{i\cdot}$ satisfying (\ref{varInequa_j}) is an optimal solution of (\ref{objective_C}). It remains to show that (\ref{varInequa_j}) is equivalent to (\ref{sys_varInequa}). If (\ref{varInequa_j}) is true for all $i$, then (\ref{sys_varInequa}) follows immediately for $\bm{\beta}^{\text{BR}}$. If (\ref{sys_varInequa}) is true, one can always take a specific profile $\bm{\beta}$ such that $\bm{\beta}_{k\cdot}=\bm{\beta}^\prime_{k\cdot}$ for all $k\neq i$ to obtain $\langle\tilde{\bm{c}}_{i\cdot} (\bm{\beta}^{\prime}_{i\cdot}), \bm{\beta}_{i\cdot}-\bm{\beta}^\prime_{i\cdot}\rangle\geq 0$ for SBS $i$, which means $\bm{\beta}^\prime_{i\cdot}$ is one of the best response solutions. According to Definition \ref{def_NE}, we can conclude that $\bm{\beta}^\prime_{i\cdot}$ is a Nash equilibrium.
\end{proof}

\subsection{Algorithm for Achieving Nash Equilibrium}
In this subsection, we first present the \emph{best-response} algorithm which is used to obtain $\bm{\beta}^{\text{BR}}_{i\cdot}$ for each SBS. Then all SBSs take turns in a round-robin fashion to perform the best-response until an Nash equilibrium $\bm{\beta}^{\text{NE}}$ is reached.
\begin{figure}[htb]
	\centering	
	\includegraphics [width=0.5 \linewidth]{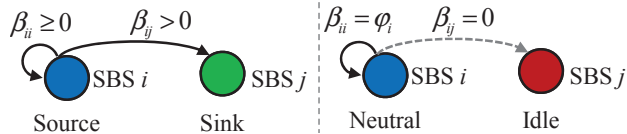}
	\vspace{-0.1 in}
	\caption{Illustration of SBS categories with respect to SBS $i$.}
	\label{category}
\end{figure}

Analogous to the previous section, SBSs are classified into different categories (see illustration in Fig.\ref{category}):
\begin{itemize}
    \item \textbf{Source SBS} ($\mathcal{R}$). A SBS is a source SBS if it offloads a positive portion of its \textit{pre-offloading} workloads to other SBSs and processes the rest of workloads locally ($0\leq \beta_{ii} < \phi_i$).
    \item \textbf{Neutral SBS} ($\mathcal{U}$). A SBS is a neutral SBS if it processes all its \textit{pre-offloading} workloads locally ($\beta_{ii}=\phi_i$).
    \item \textbf{Idle SBS with respect to SBS $i$} ($\mathcal{I}_i$). A SBS is an idle SBS with respect to SBS $i$ if it does not receive any workload from SBS $i$ ($\beta_{ij}=0$).
    \item \textbf{Sink SBS with respect to SBS $i$} ($\mathcal{S}_i$). A SBS is a sink SBS with respect to SBS $i$ if it receives workload from SBS $i$ ($\beta_{ij} > 0$).
\end{itemize}
The first two categories are defined depending on how the SBS handles its own \textit{pre-offloading} workloads. The last two categories are defined depending on how the SBS handles other SBSs' workloads. Since there are $N-1$ other SBSs, the categories are defined with respect to each SBS. Unlike the categories in Section IV, these categories are not mutually exclusive and hence, a SBS can belong to multiple categories at the same time.

To solve the best-response problem  in (\ref{objective_C}) for each SBS $i$, we define two auxiliary functions.
\begin{definition}
Define  $d_{ij}(\beta_{ij})\triangleq \frac{\partial[\beta_{ij}D^f_j(\beta_{ij})]}{\partial \beta_{ij}} = \frac{\mu_{ij}}{(\mu_{ij}-\beta_{ij})^2}$ as the pair-specific marginal computation delay function; $g_i(\lambda_i) \triangleq \frac{\partial[\lambda_iD^g(\lambda_i)]}{\partial\lambda_i}=\frac{\tau \Lambda_{-i}}{(\Lambda_{-i}-\tau\lambda_i)^2}$ as the pair-specific marginal congestion delay function, where $\mu_{ij}=\mu_j-\sum_{k=1,k\neq i}^{N}\beta_{kj}$ and $\Lambda_{-i}=1-\tau\sum_{k=1,k\neq i}^{N}\lambda_k$.
\end{definition}

The \textit{pre-offloading} Pair-specific Marginal Computation Cost (PMaCC) is therefore
\begin{equation}\label{PMaCC}
\xi_{ij} =\left\{
\begin{split}
&Vd_{ij}(0), & j\neq i\\
&Vd_{ii}(\phi_i)+\kappa q_i, & j=i
\end{split}\right.,
\end{equation}
by initializing $\beta_{ii}=\phi_i$ and $\beta_{ij}=0, \forall j\in\mathcal{N}, j\neq i$.

\begin{theorem}\label{noncoop_policy}
	In the SBS peer offloading game, the category that SBS $i$ belongs to and its best response peer offloading strategy $\bm{\beta}^{\text{BR}}_{i\cdot}$ can be decided as follow:\\
	For SBS i itself:\\
	(a) If $\xi_{ii}>\alpha_i+Vg_i(\lambda^{\text{BR}}_i)$, then $i\in \mathcal{R}$ and$\beta^{\text{BR}}_{ii}=\left[d_{ii}^{-1}\left(\frac{1}{V}(\alpha_i+Vg_i(\lambda^{\text{BR}}_i)-\kappa q_i)\right)\right]^+$;\\
	(b) If $\xi_{ii}\leq \alpha_i+Vg_i(\lambda^{\text{BR}}_i)$, then $i\in \mathcal{U}$ and $\beta^{\text{BR}}_{ii}=\phi_i$;\\
	For SBS $j$ other than $i$ ($j \neq i$):\\
	(c) If $\xi_{ij}>\alpha_i$, then $j\in \mathcal{I}_i$ and $\beta^{\text{BR}}_{ij}=0$;\\
	(d) If $\xi_{ij}<\alpha_i$, then $j\in\mathcal{S}_i$ and $\beta^{\text{BR}}_{ij}=d_{ij}^{-1}(\frac{\alpha_i}{V})$;\\
	where $\lambda^{\text{BR}}$, $\alpha_i$ are the solution to workload flow equation
	\begin{align}\label{flowconstraint}
	\underbrace{\sum_{j\in\mathcal{S}_i}d_{ij}^{-1}(\frac{\alpha_i}{V})}_{\lambda^S_i\emph{: inbound workload to } \mathcal{S}_i} =\textbf{\emph{1}}\{i\in\mathcal{R}\}\cdot\underbrace{\left(\phi_i-[d_{ij}^{-1}(\frac{1}{V}(\alpha_i+Vg_i(\lambda^{\text{BR}}_i)-q_i\kappa))]^+\right)}_{\lambda^R_i\emph{: outbound workload from SBS }i}. \nonumber
	\end{align}
\end{theorem}
\begin{proof}
	See Appendix C in Supplementary File.
\end{proof}
Theorem \ref{noncoop_policy} can be explained in a similar way as Theorem \ref{centralized_solution}. Figure \ref{noncoop} depicts the changes of marginal computation costs when SBS $i$ performs best response. One major difference is that in the best-response algorithm, SBS $i$ determines its sink SBSs by examining only the marginal computation delay cost (i.e., $d_{ij}$), regardless of the marginal energy consumption cost ($\kappa q_j$) of other SBSs. This is actually intuitive since each SBS aims to minimize its own cost rather than the overall system cost. The solution for $\alpha_i$ can be obtained by a binary search under the workload flow equation \eqref{flowconstraint} as designed by the best-response algorithm in Algorithm \ref{alg:best_response}.

\begin{algorithm}[htb]\label{alg:best_response}
	\caption{Best-response}
	\KwIn{$\phi_i$, $V$, $q_i(t)$, $\bm{\beta}_{-i\cdot}$}
	\KwOut{ optimal peer offloading strategy $\bm{\beta}^*_{i\cdot}$ for SBS $i$ }
	$\beta_{ij}\leftarrow 0 (j\neq i), \beta_{ii}\leftarrow \phi_i$\;
	Calculate $\xi_{ij}, \forall j\in\mathcal{N}$ as in \eqref{PMaCC}\;
	Find $\xi_{i,\max}=\max\limits_{j\in\mathcal{N},j\neq i} \xi_{ij}$; $\xi_{i,\min}=\min\limits_{j\in\mathcal{N},j\neq i}\xi_{ij}$\;
    \textbf{If} $\xi_{i,\min}+Vg_i(0)>\xi_{ii}$
    \textbf{STOP} (no peer offloading is required)\;
    $a\leftarrow \xi_{i,\min}$\;
    $b\leftarrow \xi_{i,\max}$\;
    \While {$|\lambda^S_i(\alpha_i)-\lambda^R_i(\alpha_i)| \geq \tilde{\epsilon}$}
    { $\lambda^S_i(\alpha_i)\leftarrow 0$,$~\lambda^R_i(\alpha_i)\leftarrow 0$\;
      $\alpha_j\leftarrow \dfrac{1}{2}(a+b)$\;
      Get in order: $\mathcal{S}_i(\alpha_i)$, $\lambda^S_i(\alpha_i)$, $\mathcal{R}(\alpha_i)$, $\mathcal{U}(\alpha_i)$, $\lambda^R_i(\alpha_i)$, $\bm{\beta}_{i\cdot}$ according to Theorem \ref{noncoop_policy}\;
      \textbf{if} {$\lambda^S_i(\alpha_i)>\lambda^R_i(\alpha_i)$}~~\textbf{then}~~$b\leftarrow\alpha_i$\;\textbf{else}~~~$a\leftarrow\alpha_i$\;
    }
    return $\bm{\beta}^{\text{BR}}_{i\cdot}$;
\end{algorithm}

\begin{figure}[b]
	\vspace{-0.1 in}
	\centering	
	\includegraphics [width=0.75 \linewidth]{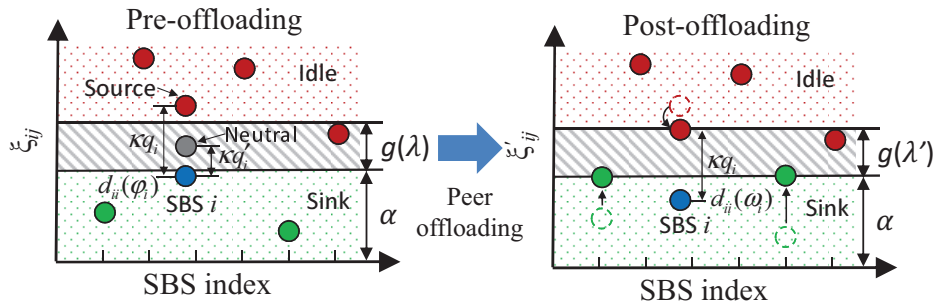}
	\caption{Changes in the marginal computation costs after best response.}
	\label{noncoop}
    
\end{figure}

With the best-response algorithm, we propose the algorithm OPEN-Autonomous to obtain the Nash equilibrium $\bm{\beta}^{\text{NE}}$, where SBSs take turns to run the best-response algorithm in a round-robin fashion. The iteration terminates if the total cost change of SBSs is less than a sufficiently small tolerance, in which case the SBS sends a terminating message to be propagated among SBSs. An important question is whether such best-response based algorithm indeed converges. There exists results about the convergence of such algorithm in the context of routing in parallel links \cite{korilis1997capacity}. For our peer offloading game there exists a unique Nash equilibrium because the cost function of the players are continuous, convex and increasing. Our simulation of in Section \ref{sec_simulation} also confirms the convergence of the best-response algorithm.

\subsection{Price of anarchy}
We now analyze the price of anarchy (PoA) of the SBS peer offloading game, which is a measure of efficiency loss due to the strategic behavior of players. Recall  that $\bm{\beta}^{*}$ is the centralized optimal solution that minimizes the system-wide cost in each time slot. Let $\mathcal{B}^{\text{NE}}$ be the set of Nash equilibria of SBS peer offloading game. The PoA is defined as:
\begin{equation*}
\text{PoA}\triangleq\dfrac{\max_{\bm{\beta}\in\mathcal{B}^{\text{NE}}}C(\bm{\beta})}{C(\bm{\beta}^*)}=\dfrac{\max_{\bm{\beta}\in\mathcal{B}^{\text{NE}}}\sum_{i\in\mathcal{N}}C_i(\bm{\beta}_{i\cdot})}{\sum_{i\in\mathcal{N}}C_i(\bm{\beta}^{*}_{i\cdot})}.
\end{equation*}

In the following, we give a general bound on the PoA.
\begin{theorem}[Bound of PoA] \label{theorm:poa}
	Let $\bm{\beta}^{\text{NE}}$ be a Nash equilibrium and $\bm{\beta}^{*}$ be the optimal peer offloading profile for the per-slot problem \textbf{P2}. Then the PoA of the non-cooperative SBS peer offloading game satisfies:
	\begin{equation*}
		1\leq\text{\emph{PoA}}\leq\dfrac{1}{1-\rho(\bm{c})},
	\end{equation*}
where $\rho(\bm{c}) \triangleq \sup_{j\in\mathcal{N}}\rho(\bm{c}_{\cdot j})$, and

\begin{equation*}
\begin{split}
&\rho(\bm{c}_{\cdot j}) \triangleq \sup_{\bm{\beta},\hat{\bm{\beta}}\in\mathcal{B}} \frac{1}{\sum_{i=1}^{N}\beta_{ij}c_{ij}(\bm{\beta}_{i\cdot})} \cdot \sum_{i=1}^{N}\left[(\tilde{c}_{ij}(\bm{\beta}_{i\cdot})-c_{ij}(\hat{\bm{\beta}}_{i\cdot}))\hat{\beta}_{ij} + \left( c_{ij}(\bm{\beta}_{i\cdot})-\tilde{c}_{ij}(\bm{\beta}_{i\cdot}) \right)\beta_{ij}\right] .
\end{split}
\end{equation*}
\end{theorem}
\begin{proof}
	We rearrange system-wide cost $C(\bm{\beta}^{\text{NE}})$ as follows:
	\begin{equation*}
	\begin{split}
		&C(\bm{\beta}^{\text{NE}})\\
		&=\sum_{i\in\mathcal{N}}\sum_{j\in\mathcal{N}}\left[\left(c_{ij}(\bm{\beta}^{\text{NE}})-\tilde{c}_{ij}(\bm{\beta}^{\text{NE}})\right)\beta^{\text{NE}}_{ij}+\tilde{c}_{ij}(\bm{\beta}^{\text{NE}})\beta^{\text{NE}}_{ij}\right]\\
		&\stackrel{(\dag)}{\leq}\sum_{i\in\mathcal{N}}\sum_{j\in\mathcal{N}}\left(c_{ij}(\bm{\beta}^{\text{NE}})-\tilde{c}_{ij}(\bm{\beta}^{\text{NE}})\right)\beta^{\text{NE}}_{ij}+\sum_{i\in\mathcal{N}}\sum_{j\in\mathcal{N}}\tilde{c}_{ij}(\bm{\beta}^{\text{NE}})\hat{\beta}_{ij}\\
		&\stackrel{(\ddag)}{\leq}\sum_{j\in\mathcal{N}}\left[\rho(\bm{c}_{\cdot j})\sum_{i\in\mathcal{N}}\beta^{\text{NE}}_{ij}c_{ij}(\bm{\beta}^{\text{NE}})\right]+C(\hat{\bm{\beta}})\\
		&\leq\rho(\bm{c})C(\bm{\beta}^{\text{NE}})+C(\hat{\bm{\beta}}).
	\end{split}
	\end{equation*}
	where we use variational inequality \eqref{sys_varInequa} to get the inequality ($\dag$) and the definition of $\rho(\cdot)$ to get the inequality ($\ddag$). We finish by taking $\hat{\bm{\beta}}=\bm{\beta}^{*}$.
\end{proof}

We further establish an upper bound on $\rho(\bm{c}_{\cdot j})$, which is given by the following Lemma.
\begin{lemma}\label{lemma:rho}
	Define $\delta(\bm{c}_{\cdot j})=\sup_{i,k\in\mathcal{N}}\frac{c_{ij}^\prime(\bm{\beta}_{i\cdot})}{c_{kj}^\prime(\bm{\beta}_{k\cdot})}$ and $\eta(\bm{c}_{\cdot j}) = \sup_{i,k\in\mathcal{N}}\frac{\omega_jc_{ij}^\prime(\bm{\beta}_{i\cdot})}{c_{kj}(\bm{\beta}_{k\cdot})}, \forall \bm{\beta}_{i\cdot}\in\mathcal{B}_{i\cdot},\forall \bm{\beta}_{k\cdot}\in\mathcal{B}_{k\cdot}$, where $c_{ij}^{\prime}(\bm{\beta}_{i\cdot})=\frac{\partial c_{ij}(\bm{\beta}_{i\cdot})}{\partial\beta_{ij}}$. If each $c_{ij}$ is differentiable, nonnegative, increasing, and convex, then the following inequality holds
	\begin{equation*}
		\rho(\bm{c}_{\cdot j})\leq\dfrac{\eta(\bm{c}_{\cdot j})}{3+\dfrac{4}{\delta(\bm{c}_{\cdot j})(N-1)}}.
	\end{equation*}
\end{lemma}

\begin{proof}
	The proof follows Proposition 5.5 in \cite{durr2014congestion}.
\end{proof}

According to Theorem \ref{theorm:poa} and Lemma \ref{lemma:rho}, we see that the bound of PoA is mainly decided by $\delta(\bm{c}_{\cdot j})=\sup_{i,k\in\mathcal{N}}\frac{c_{ij}^\prime(\bm{\beta}_{i\cdot})}{c_{kj}^\prime(\bm{\beta}_{k\cdot})}$, i.e., the maximal ratio of SBSs' pair-specific marginal cost values with respect to SBS $j$. A larger $\delta(\bm{c}_{\cdot j})$ leads to a larger $\rho(\bm{c}_{\cdot j})$, consequently, a larger PoA. The $\delta(\bm{c}_{\cdot j})$ grows with the heterogeneity levels in task arrival rates and computation capacity among SBSs. For example, suppose we have a SBS with an extremely large service rate, which gives a very large $\delta(\bm{c}_{\cdot j})$. Then, all other SBSs tend to offload workload to that SBS, which causes a large congestion delay in the LAN and hence increases the PoA value. Theoretically quantifying PoA value is hard since it heavily depends on the task arrival pattern and system configuration. In the simulation, we measure the PoA value in each time slot.

We have proved that for each time slot $t$ the PoA of the peer offloading game is bounded. Let $\varrho^{\max}$ be the maximum achievable PoA across all time slots. Then, we have:
\begin{align}\label{max_PoA}
&\sum_{i=1}^{N}\left(V\cdot D^t_i(\bm{\beta}^{\text{NE},t})+q_i(t) \cdot E^t_i(\bm{\beta}^{\text{NE},t})\right)
\leq &\varrho^{\max}\sum_{i=1}^{N}\left(V\cdot D^t_i(\bm{\beta}^{*,t})+q_i(t) \cdot E^t_i(\bm{\beta}^{*,t})\right) .
\end{align}
The strategic behavior of SBSs in the peer offloading game cause performance loss in both delay and energy efficiency. The following theorem rigorously shows the performance guarantee of OPEN-Autonomous.  
\begin{theorem}\label{Theorem: performance_OPEN_aut}
Following the peer offloading decision $\bm{\beta}^{\text{NE},t}$ obtained by OPEN-Autonomous, the long-term system delay cost satisfies:
	\begin{align*}
	\lim_{T\rightarrow\infty}\dfrac{1}{T}\sum_{t=0}^{T-1}\sum_{i=1}^{N}\mathbb{E}\left\{D_i^t(\bm{\beta}^{\text{NE},t})\right\}\leq\varrho^{\max}\Psi(\frac{\varrho^{\max}-1}{\varrho^{\max}}\bar{E}^\max)+\dfrac{B}{V},
	\end{align*}
	and the long-term energy deficit of SBSs satisfies:
	\begin{align*}
	\lim_{T\rightarrow\infty} \dfrac{1}{T}\sum_{t=0}^{T-1}\sum_{i=1}^{N} \mathbb{E}\left\{E^t_i(\bm{\beta}^{\text{NE},t})-\bar{E}_i\right\} \leq \dfrac{B+V(\varrho^{\max}D^{\max}_{\text{sys}}-D^{\text{opt}}_{\text{sys}})}{\varrho^{\max} \epsilon-(\varrho^{\max}-1)\bar{E}^{\max}},
	\end{align*}
	where $\varrho^{\max}$ is the largest \emph{PoA} value; $\Psi(\frac{\varrho^{\max}-1}{\varrho^{\max}}\bar{E}^\max)$ is the delay performance achieved by stationary policy $L$ satisfying energy consumption $\mathbb{E}\{E_i(\bm{\beta}^{L,t})-\bar{E}_i\}\leq -\frac{\varrho^{\max}-1}{\varrho^{\max}}\bar{E}^\max$, where $\bar{E}^{\max}=\max_{i\in\mathcal{N}}\bar{E}_i$; $\epsilon>0$ is a constant which represents the long-term energy surplus achieved by some stationary strategy.
\end{theorem}
\begin{proof}
	See Appendix D in Supplementary File.
\end{proof}
Theorem \ref{Theorem: performance_OPEN_aut} still shows a $[O(V),O(1/V)]$ tradeoff between delay and energy deficit. However, the delay performance achieved is no longer comparable with the optimal system delay $D^{\text{opt}}_{\text{sys}}$. Instead, the bound of delay performance is defined on the stationary policies with a reduced long-term energy constrain $\lim_{T\rightarrow\infty}\frac{1}{T}\sum_{t=0}^{T}\mathbb{E}\{E_i(\bm{\beta}^{t})\}\leq \bar{E}^{\max}/\varrho^{\max}$. Notice that if $\varrho^{\max}=1$, then Theorem \ref{Theorem: performance_OPEN_aut} is identical to Theorem \ref{Theorem_Online}. 

\section{Simulation}\label{sec_simulation}
Systematic simulations are carried out to evaluate the performance of proposed algorithm under various system settings. We assume that the edge network is deployed in a commercial complex where the business tenants deploy their own edge facilities (SBS and edge server) to serve their employees. These SBSs are connected to a LAN and hence can cooperate with each other via peer-offloading. The scale of the considered commercial complex should be large ($>$100,000 square feet \cite{katipamula2012small}), such that multiple SBSs are likely to be deployed by different business tenants and the collaboration among SBSs can be exploited. We simulated a 100m$\times$100m commercial complex (107,639 square feet) served by a set of SBSs whose locations are decided by homogeneous Poisson Point Process (PPP). The main advantage of PPP is that it captures the fact that SBSs are randomly deployed by individual owners. The density of PPP process is set as $10^{-3}$. With this PPP density and commercial complex area, the expected number of SBSs generated by PPP is 10 which is also the average number of business tenants in a commercial complex within an area around 100,000 square feet \cite{IREM}. Here, we assume that on average each business tenant will deploy an SBS. In each time slot, the UEs are randomly scattered in the network. Since the average working space of a worker is 250 square feet \cite{Gridium}, the expected number of users in the commercial complex is 400. Considering the variation in the occupancy rate across the time, we assume that the number of UEs in each time slot is randomly drawn from $[200, 600]$. Each UE is randomly assigned to one of its nearby SBSs. For an arbitrary UE, its task generation follows a Poisson process with arrival rate $\pi^t_m\in[0,4]$task/sec. The expected number of CPU cycles for each task is $h=40$M. The CPU frequency of edge server at SBSs is $f_n=3$GHz. These parameters are picked to capture the fact that an edge server can be overloaded from time to time. The expected input data size of each task is $s=0.2$Mb. Therefore, for a typical 100Mb fast Ethernet LAN, the expected transmission delay for one task is $\tau= 200$ms. The channel gain $H_m^t$ for calculating the wireless transmission is modeled by the indoor path-loss: $L[\text{dB}]=20 \log(f^{\text{TX}} [\text{MHz}])+N_L  \log(d[\text{m}])-28$, where the values of the parameters are listed in Table \ref{para_set}. The wireless channel bandwidth is $W=20$MHz, the noise power is $\delta^2= -174$dBm/Hz, the transmission power of UEs is $P_m^u= 10$dBm. These parameters are picked by following a typical wireless communication setting. Consider that the energy consumption of one CPU cycles is 8.2nJ, the expected energy consumption for each time at SBS n is $\kappa_n=9\times10^{-5}$Wh and the long-term energy constraint is set as 22W⋅h per hour.
\begin{table}
		\renewcommand\arraystretch{1.2}	
		\centering		
		\caption{Simulation setup: system parameters}		
		\begin{tabular}{l|c}			
			\hline			
			Parameters & Value\\			
			\hline
			Expected number of SBSs, $N$   & 10 (Homogeneous PPP)\\
			Expected number of UEs, $M$  & 400\\			
			Task arrival rate from UE $m$, $\pi^t_m$ & $[0,4]$task/sec\\			
			Expected num. of CPU cycles per-task, $h$ & 40M \\			
			CPU frequency of edge server at SBS $n$, $f_n$ & 3GHz \\			
			Expected input data size per-task, $s$ & 0.2Mb\\			
			Transmission delay for one task in LAN, $\tau$ & 200ms\\			
			Wireless transmission frequency, $f^{\text{TX}}$ & 900MHz\\			
			Distance power loss coefficient, $N_L$  & 20\\			
			Wireless channel bandwidth, $W$ & 20MHz \\			
			Noise power, $\delta^2$ & -174dBm/Hz\\			
			Transmission power of UEs, $P^u_m$ & 10dBm\\			
			Expected energy consumption per-task, $\kappa_n$ & $9\times10^{-5}$W$\cdot$h \\			
			Long-term energy constraint, $\bar{E}_n$ & 22W$\cdot$h (per hour)\\
			\hline			
		\end{tabular}		
		\label{para_set}
	\vspace{-0.15 in}
\end{table}
The performances of OPEN-Centralized (OPEN-C) and OPEN-Autonomous (OPEN-A) are compared with three benchmarks:
\begin{itemize}
	\item \textbf{No Peer offloading (NoP)}: peer offloading among SBSs is not enabled in the network. Each SBS processes all the tasks received from the end users. Moreover, the long-term constraint is not enforced since some SBSs have to exceed energy constraint to satisfy all the tasks due to the heterogeneity in spatial task arrival.
	\item \textbf{Delay-Optimal (D-Optimal)}: we apply the method in \cite{penmatsa2011game} where SBS peer offloading is considered as a static load balancing problem aiming to achieve the lowest system delay regardless of the long-term energy constraints.
	\item \textbf{Single-Slot Constraint (SSC)}: Instead of following a long-term energy constraint, the network operator poses a hard energy constraint in each time slot, i.e. $E^t_i(\bm{\beta}^t)\leq\bar{E}_i$, such that the long-term energy constraint is satisfied.
\end{itemize}

\subsection{Run-time Performance Evaluation}
\begin{figure}[htb]
	\centering	
	\subfigure[Time average system delay.]{\label{Tave_delay}
		\includegraphics[width=0.45\linewidth]{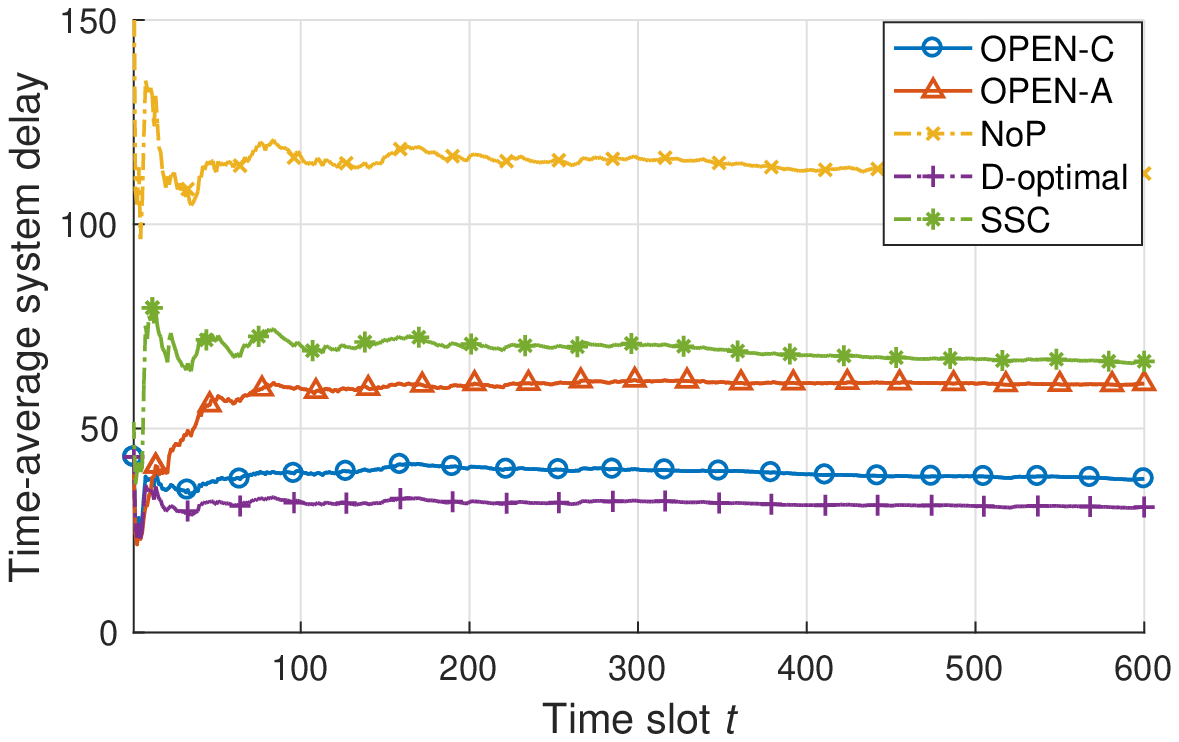}}
	\subfigure[Time average energy deficit.]{\label{Tave_deficit}
		\includegraphics[width=0.45\linewidth]{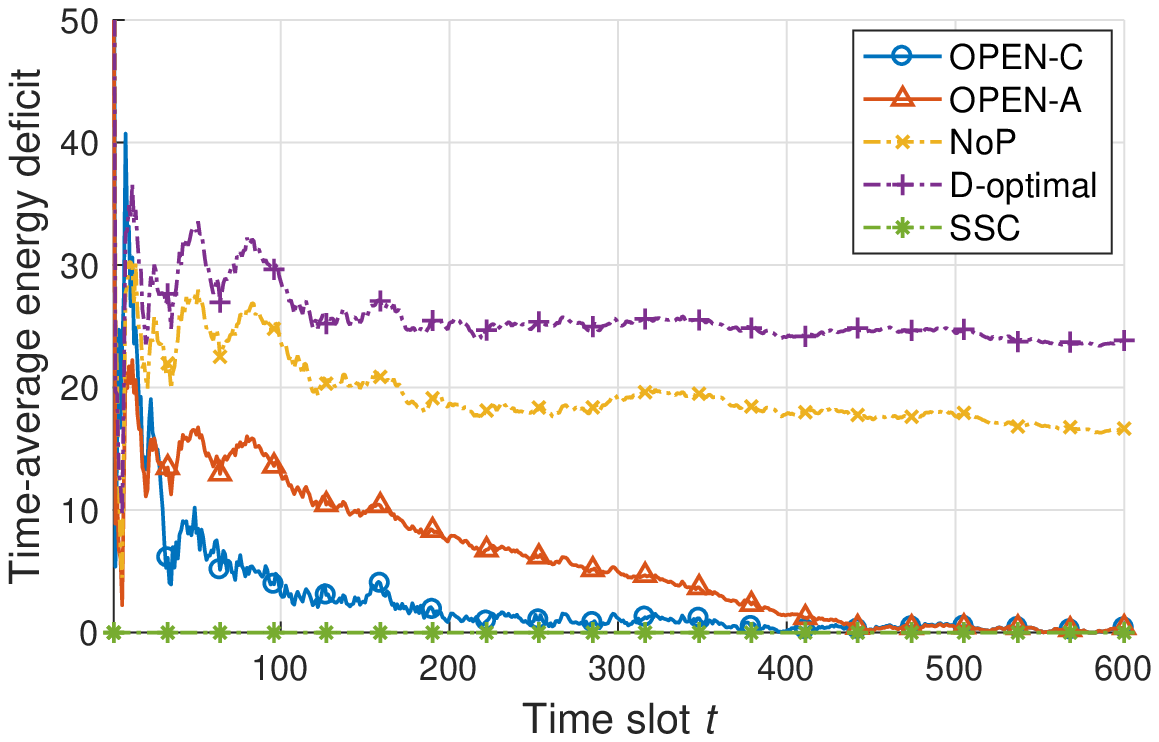}}
	\caption{Long-term performance analysis.}
	\label{fig:runtime_performance}
\end{figure}
Figure \ref{fig:runtime_performance} shows the long-term system performance obtained by running OPEN and we mainly focus on two metrics: the time-average system delay cost in Figure \ref{Tave_delay} and the time-average energy deficit in Figure \ref{Tave_deficit}. It can be observed that without peer offloading the edge system bears a high delay cost and a large energy deficit, since SBSs can be easily overloaded due to spatially and temporally heterogeneous task arrival pattern. By contrast, other three schemes with peer offloading enabled (D-Optimal, SSC, and OPEN) achieve much lower system delay cost. Specifically, D-Optimal achieves the lowest delay cost since it is designed to minimize the delay cost by fully utilizing the computation resource regardless of the energy constraints. Therefore, D-Optimal incurs a large amount of energy deficit as shown in Figure \ref{Tave_deficit}. The main purpose of OPEN is to follow the long-term energy constraint of each SBS while minimizing the system delay. As can be observed in Figure \ref{Tave_deficit}, the time-average energy deficits of both OPEN-C and OPEN-A coverage to zero, which means that the long-term energy constraints are satisfied by running OPEN. Moreover, OPEN-C achieves a close-to-optimal delay cost and OPEN-A incurs a slightly higher delay cost due to the strategic behaviors of selfish SBSs. The SSC scheme poses an energy constraint in each time slot in order to satisfy the long-term energy constraints. As a result, the energy deficit of SSC is zero across all the time slots. However, SSC makes the energy scheduling less flexible and therefore does not handle well the heterogeneity of temporal task arrival pattern and results in a large system delay cost.   

\subsection{System Dynamics}
Figure \ref{fig:sys_dyna_delay} and Figure \ref{fig:sys_dyna_energy} show the system delay cost and the energy consumption from the 500th to 550th time slot, respectively. We see that the system delay cost is mainly decided by the total task arrival rate in the network which varies across the time slots. Usually, a larger task arrival rate will result in a higher system delay cost. Figure \ref{fig:sys_dyna_energy} depicts the energy consumption of one particular SBS and the corresponding energy deficit queue in each time slot to exemplify how the energy deficit queue works to guide the energy usage. For example, from the 530th to 535th time slot, the SBS uses a large amount of energy and enlarges the energy deficit. Therefore, in the following 5 time slots, OPEN reduces the energy consumption to cut the energy deficit. In this way, the long-term energy constraints of SBSs can be satisfied.    
\begin{figure*}[htb]
	\begin{minipage}[t]{0.5\linewidth}
		\centering
		\includegraphics[width=0.95\textwidth]{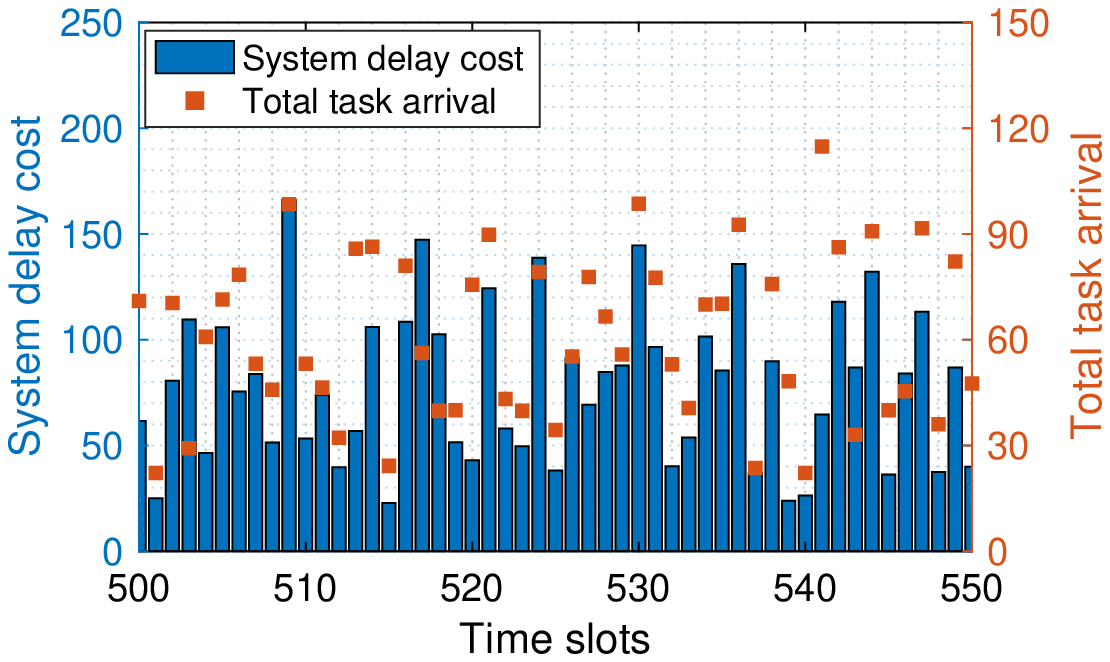}
		\vspace{-0.1 in}
		\caption{System dynamics (delay cost).}\label{fig:sys_dyna_delay}
	\end{minipage}%
	\begin{minipage}[t]{0.5\linewidth}
		\centering
		\includegraphics[width=0.95\textwidth]{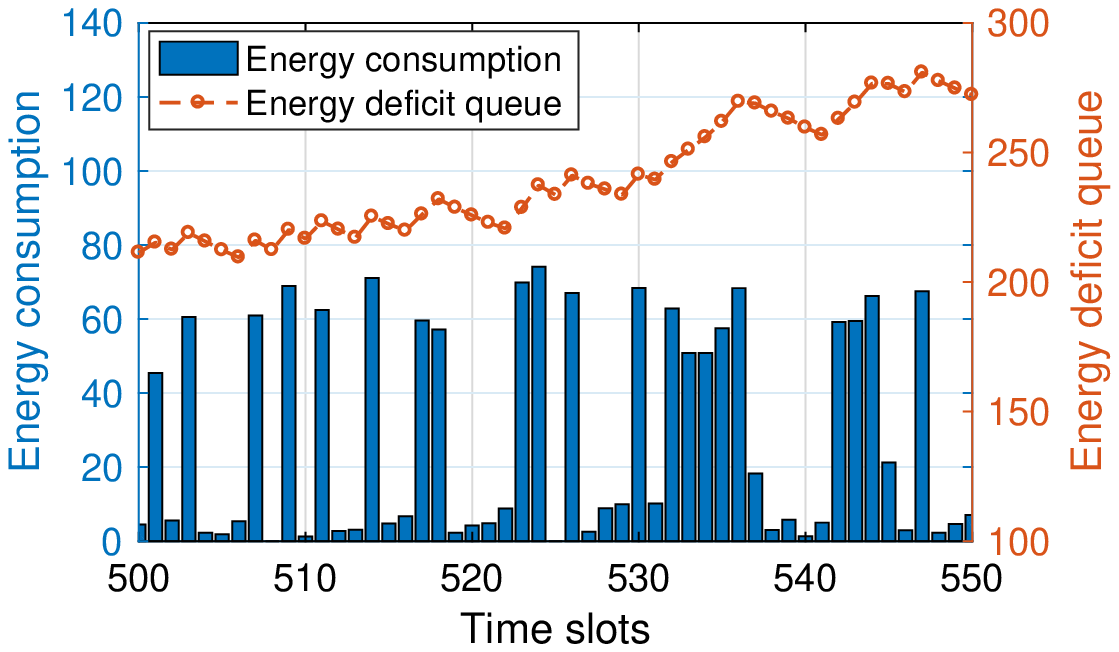}
		\vspace{-0.1 in}
		\caption{System dynamics (energy consumption).}\label{fig:sys_dyna_energy}
	\end{minipage}%
\end{figure*}

\subsection{Impact of Control Parameter $V$}
Figure \ref{fig:detradeoff} shows the impact of control parameter $V$ on the performance of OPEN. The result presents a $[O(1/V), O(V)]$ trade-off between the long-term system delay cost and the long-term energy deficit, which is consistent with our theoretical analysis. With a larger $V$, OPEN emphasizes more on the system delay cost and is less concerned with the energy deficit. As $V$ grows to the infinity, OPEN is able to achieve the optimal delay cost. It is hard to define an optimal value for $V$ since a lower system delay cost is achieved at the cost of larger energy deficit. However, it still offers a guideline for picking an appropriate $V$. In this particular simulation, the network operator is recommended to choose, for example, $V = 50$ for two reasons: ($i$) OPEN has already achieved close-to-optimal delay and little improvement is available by increasing $V$; ($ii$) the energy deficit is much smaller compared to the energy deficit achieving the optimal delay. 
\begin{figure}[htb]
	\centering
	\includegraphics[width=0.5\linewidth]{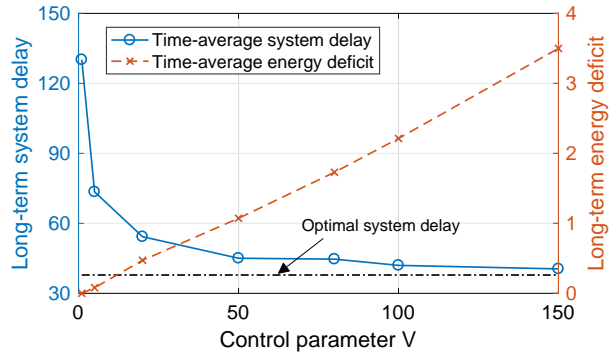}
	\vspace{-0.1 in}
	\caption{Impact of control parameter $V$.}
	\label{fig:detradeoff}
	\vspace{-0.2 in}
\end{figure}

\subsection{Composition of System  Delay}
Figure \ref{fig:delay_composition_Cop} and Figure \ref{fig:delay_composition_Aut} depict the composition of system delay cost for OPEN-C and OPEN-A, receptively. For OPEN-C, the computation delay cost takes up a large proportion of the system delay cost and the congestion delay cost is relatively small. By contrast, the congestion delay cost becomes the main part of the system delay cost in OPEN-A. The non-cooperative behavior of SBSs results in a large volume of traffic exchange in the LAN as SBSs tend to offload workload to other SBSs in order to save their own energy consumption. Note that the communication delay cost is independent of peer offloading and hence it is the same for OPEN-C and OPEN-A. 

\begin{figure*}[htb]
	\begin{minipage}[t]{0.5\linewidth}
		\centering
		\includegraphics[width=0.95\textwidth]{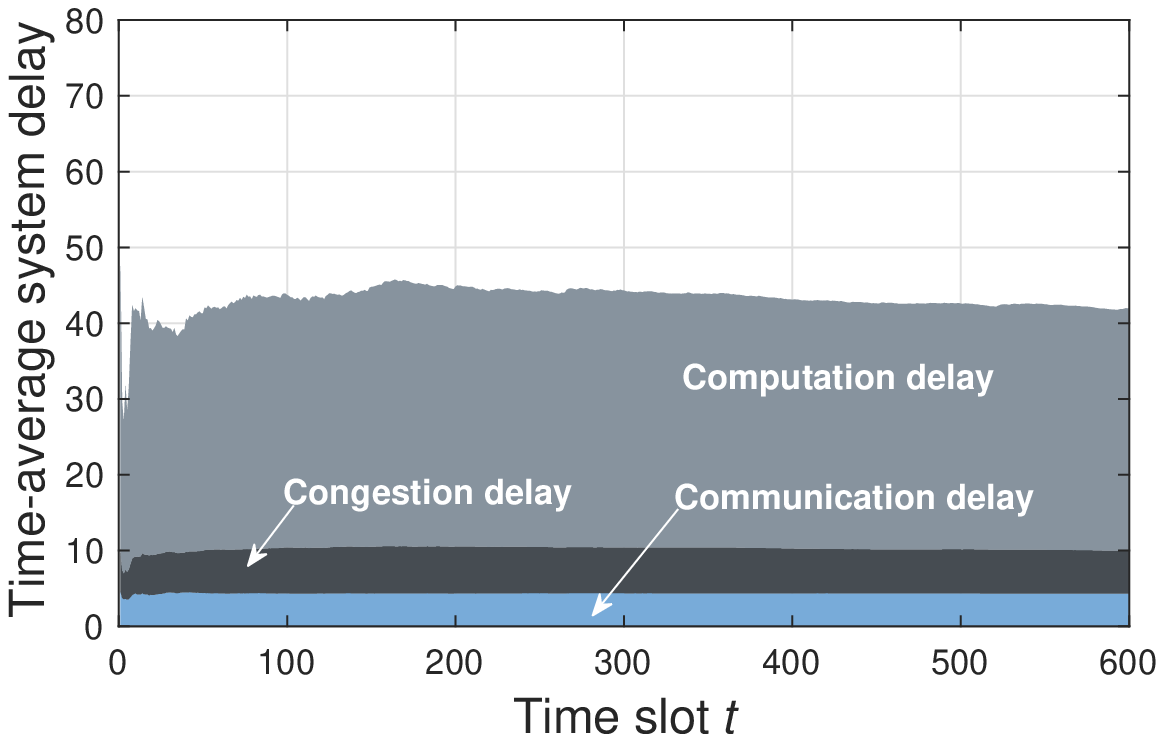}
		\vspace{-0.1 in}
		\caption{Delay cost composition of OPEN-C.}\label{fig:delay_composition_Cop}
	\end{minipage}%
	\begin{minipage}[t]{0.5\linewidth}
		\centering
		\includegraphics[width=0.95\textwidth]{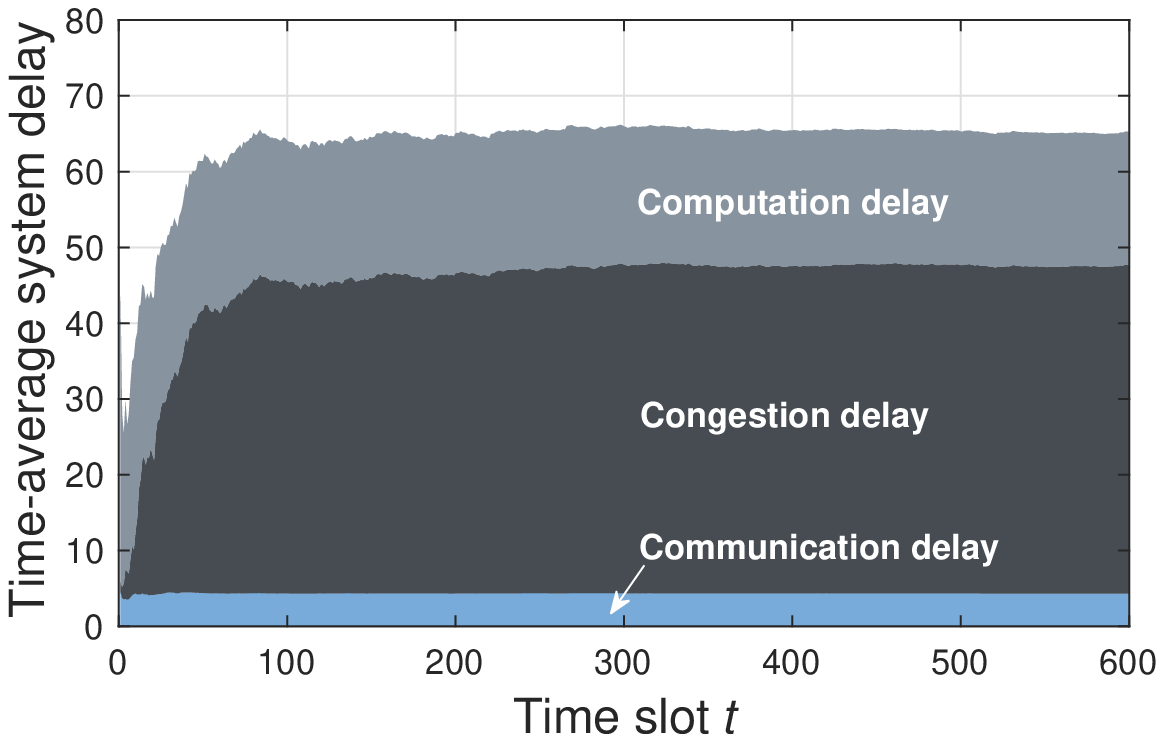}
		\vspace{-0.1 in}
		\caption{Delay cost composition of OPEN-A.}\label{fig:delay_composition_Aut}
	\end{minipage}%
\end{figure*}

\subsection{Price of Anarchy}
Since it is difficult to quantify theoretically the upper bound of PoA, we measure PoA values in the simulation. Figure \ref{fig:poa} depicts the objective value (\textbf{P2}) achieved by OPEN-C and OPEN-A from the 50th time slot to the 100th time slot from a simulation run of a total of 600 time slots. It is clearly shown that in each time slot OPEN-C achieves a strictly smaller value than OPEN-A, which means that OPEN-C outperforms OPEN-A in every time slot. For the entire 600 time slots, the mean PoA value is 1.54 and the maximum PoA value is 2.42.
\begin{figure}[htb]
	\centering
	\includegraphics[width=0.5\textwidth]{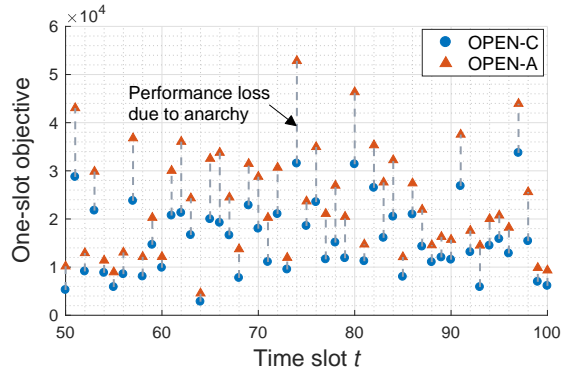}
	\vspace{-0.1 in}
	\caption{Price of anarchy.}\label{fig:poa}
\end{figure}

\subsection{System Heterogeneity}
Figure \ref{fig:spatialhetero} shows the impact of heterogeneity on the performance of OPEN. In particular, the heterogeneity of spatial task arrival pattern is considered: we regularly split the entire network into $4 \times 4$ grids and define an expected task arrival rate $\pi^{\text{grid}}_i \sim \mathcal{N}(10,\sigma^2_s), i=1,2,\dots,16$ for each grid which is drawn from a normal distribution with $\sigma_s$ being the standard deviation. The task arrival rate of UE $m$ is set as $\pi^{\text{grid}}_i$ if UE $m$ belongs to grid $i$. The level of heterogeneity is varied by changing the standard deviation $\sigma_s$, which is normalized with respect to a maximum value $\sigma_{s,\max}$ as $\sigma_s/\sigma_{s,\max}$. The result in Figure \ref{fig:spatialhetero} shows that for both OPEN-C and OPEN-A, a better system performance in terms of reduced delay cost is achieved with a higher heterogeneity level. This is because (some) SBSs are more likely to be overloaded with a higher heterogeneity level and therefore OPEN can better help to reduce the system delay by balancing the workload. 

\begin{figure}[htb]
	\centering
	\includegraphics[width=0.5\linewidth]{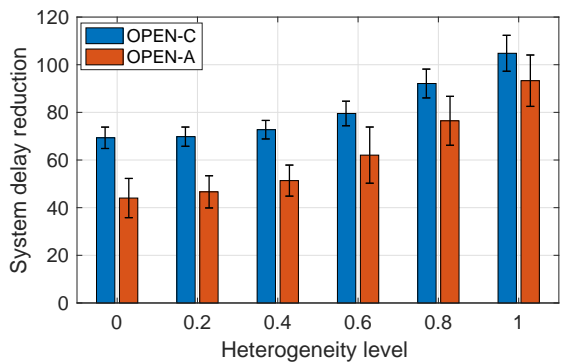}
	\caption{Impact of system heterogeneity.}
	\label{fig:spatialhetero}
\end{figure}

\subsection{Time Complexity}
Since it also takes time for OPEN to derive peer offloading decisions, especially for OPEN-A, we measured the runtime of OPEN on a typical PC to see its practical overhead. The simulation is run on a DELL PRECISION T3600 workstation with Intel(R) Xeon(R) CPU 2.8GHz and the results are presented in TABLE \ref{tab:runtime}. OPEN-C incurs an extremely low overhead, taking only 0.83ms on average (with standard deviation 0.57ms) to derive solution. By contrast, OPEN-A needs much longer time, namely 28.1ms, to obtain the Nash equilibrium since the SBSs have to take turns to run the best-response algorithm. However, the runtime for both OPEN-C and OPEN-A is negligible compared to the 1-minute peer offloading decision cycle (i.e. duration of one time slot). In addition, the information exchange is necessary to run OPEN: for OPEN-C, the centralized controller collects the information from each SBS at the beginning of each decision cycle; for OPEN-A, SBSs need to exchange the peer offloading decisions immediately after executing the best-response algorithm. TABLE \ref{tab:runtime} also shows the estimation of delay incurred by information exchange. Assume the bandwidth of the LAN is 20Mbps and the size of a message packet is 800 Bytes, the delays of information exchange incurred by OPEN-C and OPEN-A are approximately 0.2ms and 8ms, which are also negligible compared to the peer offloading decision cycle.
\begin{table}[htb]
	\centering
	\caption{Analysis of Algorithm Runtime}\label{tab:runtime}
	\begin{tabular}{l|cc} 
		\hline
		\textbf{Algorithm} & \textbf{OPEN-C} &  \textbf{OPEN-A} \\
		\hline
		\textbf{Algorithm runtime}      & 0.83 $\pm$ 0.57 ms          &  28.1 $\pm$ 4.84 ms        \\
		\textbf{Information exchange}  & $\approx$ 0.2 ms  & $\approx$ 8 ms \\
		\textbf{Computation delay (per task)} & 1.25 $\pm$ 0.38 sec   & 2.07 $\pm$ 0.65 sec\\
		\textbf{Peer offloading decision cycle}    & 1 min & 1 min  \\
		\hline
	\end{tabular}\\
	\vspace{-0.1 in}
\end{table}

\subsection{Impact of task arrival realization}
Notice that the task arrival pattern in the real-world system may not follow the assumed Poisson process. To analyze the practicality of OPEN, we implement it with different task arrival realizations. Fig. \ref{fig:task_arrival} shows the performances achieved by OPEN and NoP with two task arrival realizations: bursty arrival and non-i.i.d. arrival. We can see from Fig.\ref{delay_burst} that OPEN and NoP both have a higher delay cost in the bursty arrival case compared to that in the Poisson arrival case. This is because the tasks are more likely to queue up at edge servers when bursts occur. However, we see that the proposed algorithm still provides a 55.0\% delay reduction with bursty arrival which is similar to that of Poisson arrival. In the non-i.i.d. case, we use a Markov process to model a task arrival pattern where the intervals of task arrivals are determined by certain transition probabilities. We see from Fig. \ref{delay_noniid} that the delay reduction achieved by OPEN in the non-i.i.d case slightly decreases compared to that in the Poisson arrival case. However, applying OPEN still offers an obvious delay reduction, 37.5\%, for the edge system. These two examples indicate that the proposed algorithm can offer considerable performance improvement for the edge system even if the real task arrival does not closely follow the Poisson process. 

\begin{figure}[htb]
	\centering	
	\subfigure[Bursty arrival]{\label{delay_burst}
		\includegraphics[width=0.45\linewidth]{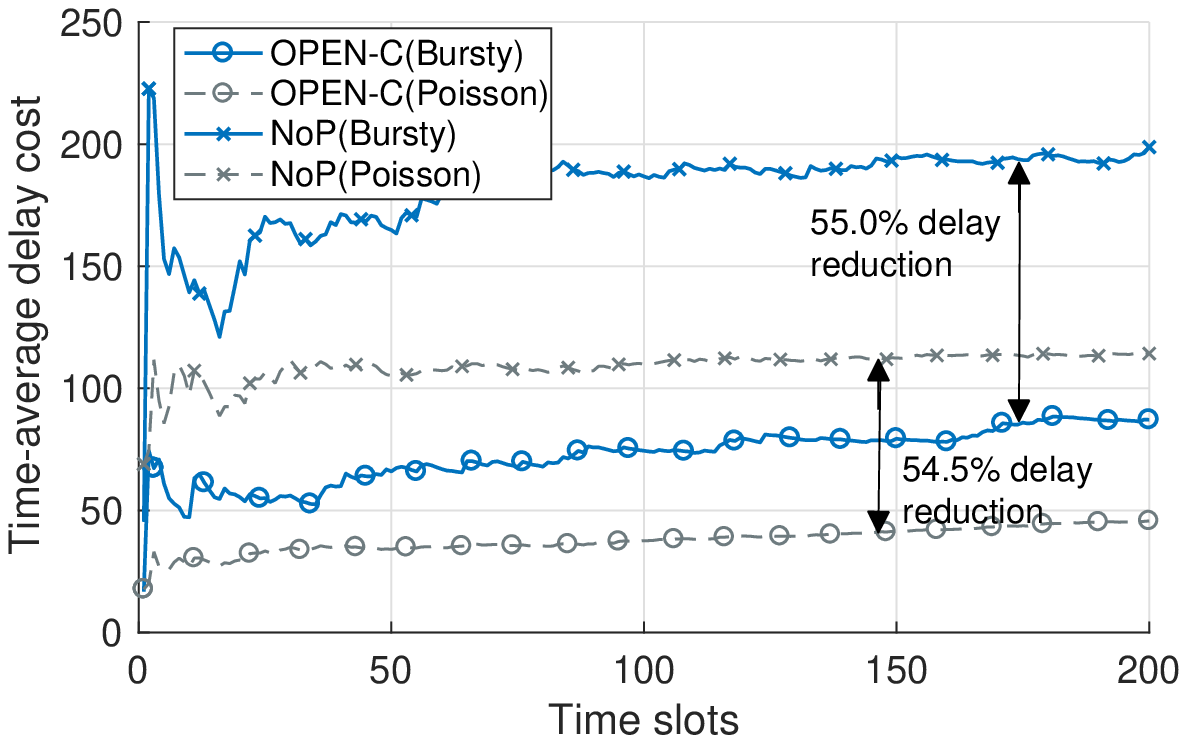}}
	\subfigure[Non-i.i.d. arrival]{\label{delay_noniid}
		\includegraphics[width=0.45\linewidth]{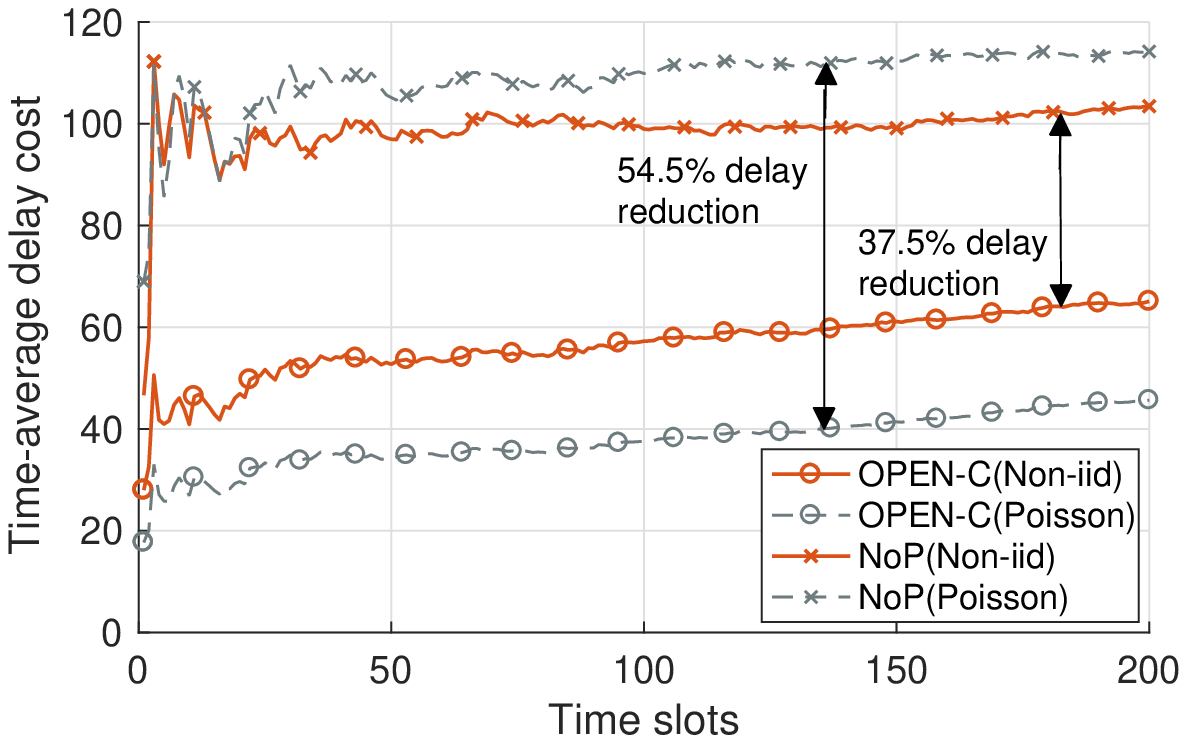}}
	\caption{Impact of task arrival realization.}
	\vspace{-0.2 in}
	\label{fig:task_arrival}
\end{figure}

\section{Conclusion}\label{sec_conclusion}
In this paper, we investigated peer offloading schemes in MEC-enabled small cell networks where heterogeneous task arrival pattern in both spatial and temporal domains is considered. We developed OPEN, a novel online peer offloading framework to optimize the edge computing performance under limited energy resources committed by individual SBSs without requiring information on future system dynamics. The proposed framework allows both centralized and autonomous decision making, and provide provable performance guarantee. We also showed that OPEN incurs acceptable overhead in practice. However, there are still issues that need to be addressed. In the current system model, we considered a simple structure of LAN to model the congestion delay during peer offloading. How to extend our analysis to more sophisticated and practical congestion scenarios needs further investigation. Moreover, the performance guarantee of OPEN rests on the assumption that the observations of task arrival rates in the current slot are precise, which may not hold for all network systems. Therefore, future efforts are needed to take into consideration of imprecise estimation of task arrival.

\bibliographystyle{IEEEtran}
\bibliography{ref}

\begin{thebibliography}{10}
\providecommand{\url}[1]{#1}
\csname url@samestyle\endcsname
\providecommand{\newblock}{\relax}
\providecommand{\bibinfo}[2]{#2}
\providecommand{\BIBentrySTDinterwordspacing}{\spaceskip=0pt\relax}
\providecommand{\BIBentryALTinterwordstretchfactor}{4}
\providecommand{\BIBentryALTinterwordspacing}{\spaceskip=\fontdimen2\font plus
\BIBentryALTinterwordstretchfactor\fontdimen3\font minus
  \fontdimen4\font\relax}
\providecommand{\BIBforeignlanguage}[2]{{%
\expandafter\ifx\csname l@#1\endcsname\relax
\typeout{** WARNING: IEEEtran.bst: No hyphenation pattern has been}%
\typeout{** loaded for the language `#1'. Using the pattern for}%
\typeout{** the default language instead.}%
\else
\language=\csname l@#1\endcsname
\fi
#2}}
\providecommand{\BIBdecl}{\relax}
\BIBdecl

\bibitem{mao2017mobile}
Y.~Mao, C.~You, J.~Zhang, K.~Huang, and K.~B. Letaief, ``A survey on mobile
  edge computing: The communication perspective,'' \emph{IEEE Communications
  Surveys Tutorials}, vol.~PP, no.~99, pp. 1--1, 2017.

\bibitem{shi2016edge}
W.~Shi, J.~Cao, Q.~Zhang, Y.~Li, and L.~Xu, ``Edge computing: Vision and
  challenges,'' \emph{IEEE Internet of Things Journal}, vol.~3, no.~5, pp.
  637--646, 2016.

\bibitem{roman2016mobile}
R.~Roman, J.~Lopez, and M.~Mambo, ``Mobile edge computing, fog et al.: A survey
  and analysis of security threats and challenges,'' \emph{Future Generation
  Computer Systems}, 2016.

\bibitem{rivera2014gartner}
J.~Rivera and R.~van~der Meulen, ``Gartner says the internet of things will
  transform the data center,'' \emph{Retrieved August}, vol.~5, p. 2014, 2014.

\bibitem{neely2010stochastic}
M.~J. Neely, ``Stochastic network optimization with application to
  communication and queueing systems,'' \emph{Synthesis Lectures on
  Communication Networks}, vol.~3, no.~1, pp. 1--211, 2010.

\bibitem{abdelnasser2014clustering}
A.~Abdelnasser, E.~Hossain, and D.~I. Kim, ``Clustering and resource allocation
  for dense femtocells in a two-tier cellular ofdma network,'' \emph{IEEE
  Transactions on Wireless Communications}, vol.~13, no.~3, pp. 1628--1641,
  March 2014.

\bibitem{Guruacharya2013Dynamic}
S.~Guruacharya, D.~Niyato, M.~Bennis, and D.~I. Kim, ``Dynamic coalition
  formation for network mimo in small cell networks,'' \emph{IEEE Transactions
  on Wireless Communications}, vol.~12, no.~10, pp. 5360--5372, October 2013.

\bibitem{sardellitti2015joint}
S.~Sardellitti, G.~Scutari, and S.~Barbarossa, ``Joint optimization of radio
  and computational resources for multicell mobile-edge computing,'' \emph{IEEE
  Transactions on Signal and Information Processing over Networks}, vol.~1,
  no.~2, pp. 89--103, June 2015.

\bibitem{queis2015small}
J.~Oueis, E.~C. Strinati, S.~Sardellitti, and S.~Barbarossa, ``Small cell
  clustering for efficient distributed fog computing: A multi-user case,'' in
  \emph{2015 IEEE 82nd Vehicular Technology Conference (VTC2015-Fall)}, Sept
  2015, pp. 1--5.

\bibitem{queis2015fogbalancing}
J.~Oueis, E.~C. Strinati, and S.~Barbarossa, ``The fog balancing: Load
  distribution for small cell cloud computing,'' in \emph{2015 IEEE 81st
  Vehicular Technology Conference (VTC Spring)}, May 2015, pp. 1--6.

\bibitem{islam2015water}
M.~Islam, S.~Ren, G.~Quan, M.~Shakir, and A.~Vasilakos, ``Water-constrained
  geographic load balancing in data centers,'' \emph{IEEE Transactions on Cloud
  Computing}, vol.~PP, no.~99, pp. 1--1, 2015.

\bibitem{liu2015greening}
Z.~Liu, M.~Lin, A.~Wierman, S.~Low, and L.~L.~H. Andrew, ``Greening
  geographical load balancing,'' \emph{IEEE/ACM Transactions on Networking},
  vol.~23, no.~2, pp. 657--671, April 2015.

\bibitem{zhang2013dynamic}
Q.~Zhang, Q.~Zhu, M.~F. Zhani, R.~Boutaba, and J.~L. Hellerstein, ``Dynamic
  service placement in geographically distributed clouds,'' \emph{IEEE Journal
  on Selected Areas in Communications}, vol.~31, no.~12, pp. 762--772, December
  2013.

\bibitem{chen2016efficient}
X.~Chen, L.~Jiao, W.~Li, and X.~Fu, ``Efficient multi-user computation
  offloading for mobile-edge cloud computing,'' \emph{IEEE/ACM Transactions on
  Networking}, vol.~24, no.~5, pp. 2795--2808, 2016.

\bibitem{chen2015decentralized}
X.~Chen, ``Decentralized computation offloading game for mobile cloud
  computing,'' \emph{IEEE Transactions on Parallel and Distributed Systems},
  vol.~26, no.~4, pp. 974--983, 2015.

\bibitem{fernando2013mobile}
N.~Fernando, S.~W. Loke, and W.~Rahayu, ``Mobile cloud computing: A survey,''
  \emph{Future generation computer systems}, vol.~29, no.~1, pp. 84--106, 2013.

\bibitem{satyanarayanan2010mobile}
M.~Satyanarayanan, ``Mobile computing: the next decade,'' in \emph{Proceedings
  of the 1st ACM workshop on mobile cloud computing \& services: social
  networks and beyond}.\hskip 1em plus 0.5em minus 0.4em\relax ACM, 2010, p.~5.

\bibitem{bonomi2012fog}
F.~Bonomi, R.~Milito, J.~Zhu, and S.~Addepalli, ``Fog computing and its role in
  the internet of things,'' in \emph{Proceedings of the first edition of the
  MCC workshop on Mobile cloud computing}.\hskip 1em plus 0.5em minus
  0.4em\relax ACM, 2012, pp. 13--16.

\bibitem{satyanarayanan2009cloudlet}
M.~Satyanarayanan, P.~Bahl, R.~Caceres, and N.~Davies, ``The case for vm-based
  cloudlets in mobile computing,'' \emph{IEEE pervasive Computing}, vol.~8,
  no.~4, 2009.

\bibitem{greenberg2008cost}
A.~Greenberg, J.~Hamilton, D.~A. Maltz, and P.~Patel, ``The cost of a cloud:
  research problems in data center networks,'' \emph{ACM SIGCOMM computer
  communication review}, vol.~39, no.~1, pp. 68--73, 2008.

\bibitem{TROPIC}
\BIBentryALTinterwordspacing
TROPIC. (2015) Distributed computing storage and radio resource allocation over
  cooperative femtocells. [Online]. Available: \url{http://www.ict-tropic.eu.}
\BIBentrySTDinterwordspacing

\bibitem{rubio2014association}
J.~Rubio, A.~Pascual-Iserte, J.~del Olmo, and J.~Vidal, ``User association for
  load balancing in heterogeneous networks powered with energy harvesting
  sources,'' in \emph{2014 IEEE Globecom Workshops (GC Wkshps)}, Dec 2014, pp.
  1248--1253.

\bibitem{tam2017joint}
H.~H.~M. Tam, H.~D. Tuan, D.~T. Ngo, T.~Q. Duong, and H.~V. Poor, ``Joint load
  balancing and interference management for small-cell heterogeneous networks
  with limited backhaul capacity,'' \emph{IEEE Transactions on Wireless
  Communications}, vol.~16, no.~2, pp. 872--884, Feb 2017.

\bibitem{lin2012online}
M.~Lin, Z.~Liu, A.~Wierman, and L.~L. Andrew, ``Online algorithms for
  geographical load balancing,'' in \emph{Green Computing Conference (IGCC),
  2012 International}.\hskip 1em plus 0.5em minus 0.4em\relax IEEE, 2012, pp.
  1--10.

\bibitem{xu2015temperature}
H.~Xu, C.~Feng, and B.~Li, ``Temperature aware workload managementin
  geo-distributed data centers,'' \emph{IEEE Transactions on Parallel and
  Distributed Systems}, vol.~26, no.~6, pp. 1743--1753, 2015.

\bibitem{lou2015spatio}
J.~Luo, L.~Rao, and X.~Liu, ``Spatio-temporal load balancing for energy cost
  optimization in distributed internet data centers,'' \emph{IEEE Transactions
  on Cloud Computing}, vol.~3, no.~3, pp. 387--397, July 2015.

\bibitem{liu2011greening}
Z.~Liu, M.~Lin, A.~Wierman, S.~H. Low, and L.~L. Andrew, ``Greening
  geographical load balancing,'' in \emph{Proceedings of the ACM SIGMETRICS
  joint international conference on Measurement and modeling of computer
  systems}.\hskip 1em plus 0.5em minus 0.4em\relax ACM, 2011, pp. 233--244.

\bibitem{farahnakian2014energy}
F.~Farahnakian, P.~Liljeberg, and J.~Plosila, ``Energy-efficient virtual
  machines consolidation in cloud data centers using reinforcement learning,''
  in \emph{Parallel, Distributed and Network-Based Processing (PDP), 2014 22nd
  Euromicro International Conference on}.\hskip 1em plus 0.5em minus
  0.4em\relax IEEE, 2014, pp. 500--507.

\bibitem{zhou2017reinforcement}
X.~Zhou, K.~Wang, W.~Jia, and M.~Guo, ``Reinforcement learning-based adaptive
  resource management of differentiated services in geo-distributed data
  centers,'' in \emph{Quality of Service (IWQoS), 2017 IEEE/ACM 25th
  International Symposium on}.\hskip 1em plus 0.5em minus 0.4em\relax IEEE,
  2017, pp. 1--6.

\bibitem{nisan2007algorithmic}
N.~Nisan, T.~Roughgarden, E.~Tardos, and V.~V. Vazirani, \emph{Algorithmic game
  theory}.\hskip 1em plus 0.5em minus 0.4em\relax Cambridge University Press
  Cambridge, 2007, vol.~1.

\bibitem{kinderlehrer2000introduction}
D.~Kinderlehrer and G.~Stampacchia, \emph{An introduction to variational
  inequalities and their applications}.\hskip 1em plus 0.5em minus 0.4em\relax
  SIAM, 2000.

\bibitem{cooper1981introduction}
R.~B. Cooper, \emph{Introduction to queueing theory}.\hskip 1em plus 0.5em
  minus 0.4em\relax North Holland, 1981.

\bibitem{liu2011efficient}
C.~H. Liu, P.~Hui, J.~W. Branch, C.~Bisdikian, and B.~Yang, ``Efficient network
  management for context-aware participatory sensing,'' in \emph{Sensor, mesh
  and ad hoc commun. and networks, 2011 8th annual IEEE Conf. on}.\hskip 1em
  plus 0.5em minus 0.4em\relax IEEE, 2011, pp. 116--124.

\bibitem{Mihailescu2010oneconomic}
M.~Mihailescu and Y.~M. Teo, ``On economic and computational-efficient resource
  pricing in large distributed systems,'' in \emph{2010 10th IEEE/ACM
  International Conference on Cluster, Cloud and Grid Computing}, May 2010, pp.
  838--843.

\bibitem{korilis1997capacity}
Y.~A. Korilis, A.~A. Lazar, and A.~Orda, ``Capacity allocation under
  noncooperative routing,'' \emph{IEEE Transactions on Automatic Control},
  vol.~42, no.~3, pp. 309--325, Mar 1997.

\bibitem{durr2014congestion}
C.~D{\"u}rr, ``Congestion games with player-specific cost functions,'' Ph.D.
  dissertation, Universit{\'e} Paris-Est, 2014.

\bibitem{katipamula2012small}
S.~Katipamula, R.~M. Underhill, J.~K. Goddard, D.~J. Taasevigen, M.~Piette,
  J.~Granderson, R.~E. Brown, S.~M. Lanzisera, and T.~Kuruganti, ``Small-and
  medium-sized commercial building monitoring and controls needs: A scoping
  study,'' Pacific Northwest National Lab.(PNNL), Richland, WA (United States),
  Tech. Rep., 2012.

\bibitem{IREM}
IREM, ``Trends in office buildings operations,'' Institute of Real Estate
  Management, Tech. Rep., 2012.

\bibitem{Gridium}
\BIBentryALTinterwordspacing
Gridium. (2015) Density trends in commercial real estate. [Online]. Available:
  \url{https://gridium.com/imagine-your-building-with-35-more-people-in-it/}
\BIBentrySTDinterwordspacing

\bibitem{penmatsa2011game}
S.~Penmatsa and A.~T. Chronopoulos, ``Game-theoretic static load balancing for
  distributed systems,'' \emph{Journal of Parallel and Distributed Computing},
  vol.~71, no.~4, pp. 537--555, 2011.

\end{thebibliography}

\newpage
\appendices
\section{Proof of Theorem 1} \label{prof_theroem_centrialized}
Let $u_i$ and $v_i$ denote the inbound and outbound traffic of SBS $i$, respectively. From the balance of total traffic in the network, we have $\sum_{i=1}^{N}u_i=\sum_{i=1}^{N}v_i$. The post-offloading workload $\omega_i$ at SBS $i$ can then be written as $\omega_i=\phi_i+u_i-v_i$ and the network traffic $\lambda$ can be written as $\lambda=\sum_{i=1}^{N}v_i$. Hence the problem becomes:
\begin{subequations}
\begin{align}
&\min F(\bm{u},\bm{v})=\sum_{i=1}^{N}V[\dfrac{\phi_i+u_i-v_i}{\mu_i-(\phi_i+u_i-v_i)} +\kappa q_i(\phi_i+u_i-v_i)]+\dfrac{\tau\sum_{i=1}^{N}v_i}{1-\tau\sum_{i=1}^{N}v_i} \label{soical_overhead}\\
&\qquad\qquad\text{s.t}.\quad \omega_i=\phi_i+u_i-v_i\geq 0,~\forall i\in\mathcal{N}\\
&\qquad\qquad\qquad -\sum_{i=1}^{N}u_i+\sum_{i=1}^{N}v_i=0\\
&\qquad\qquad\qquad u_i\geq 0,~\forall i\in\mathcal{N}\\
&\qquad\qquad\qquad v_i\geq 0, ~\forall i\in\mathcal{N}
\end{align}
\end{subequations}

The objective function in (\ref{soical_overhead}) is convex and the constraints are all linear and define a convex ployhedron. This imply that the first-order Kuhn-Tucker conditions are necessary and sufficient for optimality. Let $\alpha$, $\delta_i\leq0$, $\eta_i\leq0$, $\psi_i\leq0$ denote the Lagrange multipliers. The Lagrangian is
\begin{equation*}
\begin{split}
F(\bm{u},\bm{v})+&\alpha(-\sum_{i=1}^{N}u_i+\sum\limits_{i=1}^{N}v_i)+\sum_{i=1}^{N}\delta_i(\phi_i+u_i-v_i)+\sum_{i=1}^{N}\eta_iu_i+\sum_{i=1}^{N}\psi_iv_i
\end{split}
\end{equation*}

The optimal solution satisfies the following Kuhn-Tucker conditions:
\begin{subequations}
	\begin{align}
	&\dfrac{\partial L}{\partial u_i}=Vd_i(\phi_i+u_i-v_i)+\kappa q_i-\alpha+\delta_i+\eta_i=0,~\forall i\in\mathcal{N} \label{co_partial_u} \\
	&\dfrac{\partial L}{\partial v_i}=-V(d_i(\phi_i+u_i-v_i)+\kappa q_i)+Vg(\lambda)+\alpha-\delta_i+\psi_i=0,~\forall i\in\mathcal{N} \label{co_parital_v}\\
	&\dfrac{\partial L}{\partial\alpha}=-\sum_{i=1}^{N}u_i+\sum_{i=1}^{N}v_i=0  \label{co_alpha}\\
	& \phi_i+u_i-v_i\geq 0, \delta_i(\phi_i+u_i-v_i)=0, \delta\leq 0, \forall i\in\mathcal{N} \label{co_delta}\\
	& u_i\geq0, \eta_iu_i=0, \eta_i\leq 0, \forall i\in\mathcal{N} \label{co_eta}\\
	& v_i\geq0, \psi_iv_i=0, \psi_i\leq 0, \forall i\in\mathcal{N} \label{co_psi}
	\end{align}
\end{subequations}

In the following, we find an equivalent form of above equations in terms of $\omega_i$. Adding (\ref{co_partial_u}) and (\ref{co_parital_v}) we have $-Vg(\lambda)=\eta_i+\psi_i, \forall i\in\mathcal{N}$. Since $g>0$, either $\eta_i<0$ or $\psi_i<0$ (or both). Hence, from (\ref{co_eta}) and (\ref{co_psi}), for each $i$, either $u_i=0$ or $v_i=0$ (or both).  We consider each case separately.

\begin{itemize}[leftmargin=*]
\item \emph{Case I}: $u_i=0, v_i=0$, then we have $\omega_i=\phi_i$. It follows from (\ref{co_delta}) that $\delta_i=0$. Substituting this into (\ref{co_partial_u}) and (\ref{co_parital_v}) gives
\begin{align*}
	& Vd_i(\omega_i)+\kappa q_i=\alpha-\eta_i\geq \alpha\\
	& Vd_i(\omega_i)+\kappa q_i=Vg(\lambda)+\alpha+\psi_i<\alpha+g(\lambda)
\end{align*}
From the above, we have
\begin{align}
	\alpha\leq Vd_i(\omega_i)+\kappa q_i\leq\alpha+Vg(\lambda)
\end{align}
This case corresponds to the \emph{neutral SBS}.

\item \emph{Case II}: $u_i=0, v_i>0$. This case corresponds to \emph{source SBS}. From (\ref{co_psi}) we have $\psi_i=0$ and consider the following two subcases:
\begin{itemize}
	\item \emph{Case II.1}: $v_i<\phi_i$. Then, we have $0<\omega_i<\phi_i$. It follows from (\ref{co_delta}) that $\delta_i=0$. Substituting this in to (\ref{co_partial_u}),(\ref{co_parital_v}) gives
	\begin{align}
	&Vd_i(\omega_i)+\kappa q_i=\alpha-\eta_i\geq\alpha\\
	&Vd_i(\omega_i)+\kappa q_i=\alpha+Vg(\lambda) \label{co_active}
	\end{align}
	\item \emph{Case II.2}: $v_{i}=\phi_i$. Then, we have $\omega_i=0$ and (\ref{co_partial_u}), (\ref{co_parital_v}) gives
	\begin{align*}
	&Vd_i(\omega_i)+\kappa q_i=\alpha-\delta_i-\eta_i\geq\alpha\\
	&Vd_i(\omega_i)+\kappa q_i=\alpha+Vg(\lambda)-\delta_i\geq\alpha+g(\lambda).
	\end{align*}
	Thus, we have $Vd_i(\omega_i)+\kappa q_i\geq\alpha+Vg(\lambda).$
\end{itemize}
\item \emph{Case III}: $u_i>0, v_i=0$. Then, we have $\omega_i>\phi_i$. It follows from (\ref{co_delta}) and (\ref{co_eta}) that $\delta_i=0$ and $\eta_i=0$. Substituting this in to (\ref{co_partial_u}), we have
\begin{align} \label{co_sinks}
Vd_i(\omega_i)+\kappa q_i=\alpha.
\end{align}
This case corresponds to \emph{sink SBS}.
\end{itemize}
Using (\ref{co_active}) and (\ref{co_sinks}), we have the total flow constraint:
\begin{align*}
\sum_{i\in\mathcal{S}}&(d^{-1}_i\left(\dfrac{1}{V}(\alpha-\kappa q_i)\right)-\phi_i)=\sum_{i\in\mathcal{R}}(\phi_i-[d^{-1}_i\left(\dfrac{1}{V}(\alpha+Vg(\lambda)-\kappa q_i)\right)]^+)
\end{align*}

\section{Proof of Theorem \ref{Theorem_Online}} \label{appendix_online}
To prove the time-averaged system delay bound, we first introduce the following Lemma.
\begin{lemma}\label{Lemma:stationary_policy}
	For an arbitrary $\delta>0$, there exists a stationary and randomized policy $\Pi$ for $\textbf{P2}$, which decides $\bm{\beta}^{\Pi,t}$ independent of the current queue backlogs, such that the following inequalities are satisfied:
	\begin{align*}
	&\sum_{i=1}^{N}\mathbb{E}\{D_i(\bm{\beta}^{\Pi,t})\}\leq D^{\text{opt}}_{\text{sys}}+\delta\\
	&\mathbb{E}\{E^t_i(\bm{\beta}^{\Pi,t})-\bar{E}_i\}\leq \delta
	\end{align*}
\end{lemma}
\begin{proof}
	The proof can be obtained by Theorem 4.5 in \cite{neely2010stochastic}, which is omitted for brevity.
\end{proof}

Recall that the OPEN-Centralized seeks to choose strategies $\bm{\beta}^{*,t}$ that can minimize $\textbf{P2}$ among feasible decisions including the policy in Lemma \ref{Lemma:stationary_policy} in each time slot. By applying Lemma \ref{Lemma:stationary_policy} in to the \emph{drift-plus-penalty} inequality \eqref{drift_plus_cost}, we obtain:
\begin{equation}
\begin{split}
&\Delta(\bm{q}(t))+V\sum_{i=1}^{N}\mathbb{E}\left\{D_i^t(\bm{\beta}^{*,t})|\bm{q}(t)\right\} \\&\leq B+\sum_{i=1}^{N}q_i(t)\mathbb{E}\left\{(E^t_i(\bm{\beta}^{\Pi,t})-\bar{E}_i)|\bm{q}(t) \right\}+V\sum_{i=1}^{N}\mathbb{E}\left\{D_i^t(\bm{\beta}^{\Pi,t})|\bm{q}(t)\right\}\\
&\stackrel{(\dag)}{\leq}B+\delta\sum_{i=1}^{N}q_i(t)+V(D^{\text{opt}}_{\text{sys}}+\delta)
\end{split}
\end{equation}

The inequality $(\dag)$ is because that the policy $\Pi$ is independent of the energy deficit queue. By letting $\delta$ go to zero, summing the inequality over $t\in\{0,1,\dots,T-1\}$ and then dividing the result by $T$, we have:
\begin{align}
\dfrac{1}{T}\mathbb{E}\{L(\bm{q}(t))-L(\bm{q}(0))\}+\dfrac{V}{T}\sum_{t=0}^{T-1}\sum_{i=1}^{N}\mathbb{E}\left\{D_i^t(\bm{\beta}^{*,t})\right\}\leq B+VD^{\text{opt}}_{\text{sys}}\label{xx}
\end{align}

Rearranging the terms and considering the fact that $L(\bm{q}(t))\geq 0$ and $L(\bm{q}(0))=0$ yields the bound for long-term system delay cost.

To obtain the long-term energy deficit bound, we make following assumption: there are values $\epsilon>0$ and $\Psi(\epsilon)$ and an policy $\bm{\beta}^{L,t}$ that satisfies:
\begin{equation}\label{eq:assump}
\begin{split}
&\sum_{i=1}^{N}\mathbb{E}\{D_i(\bm{\beta}^{L,t})\}=\Psi(\epsilon)\\
&\mathbb{E}\{E_i(\bm{\beta}^{L,t})-\bar{E}_i\}\leq -\epsilon
\end{split}
\end{equation}
Plugging above into inequality \eqref{drift_plus_cost}
\begin{equation*}
\Delta(\bm{q}(t))+V\sum_{i=1}^{N}\mathbb{E}\left\{D_i^t(\bm{\beta}^{*,t})\right\}\leq B+V\Psi(\epsilon)-\epsilon\sum_{i=1}^{N}q_i(t)
\end{equation*}
Summing the above over $t\in\{0,1,\dots,T-1\}$ and rearranging terms as usual yields:
\begin{equation}
\begin{split}
&\dfrac{1}{T}\sum_{t=0}^{T-1}\sum_{i=1}^{N}\mathbb{E}\{q_i(t)\}\\
&\leq \dfrac{B+V(\Psi(\epsilon)-\dfrac{1}{T}\sum\limits_{t=0}^{T-1}\sum\limits_{i=1}^{N}\mathbb{E}\left\{D_i^t(\bm{\beta}^{*,t})\right\}}{\epsilon}\\
&\leq \dfrac{B+V(D^{\max}_{\text{sys}}-D^{\text{opt}}_{\text{sys}})}{\epsilon}
\end{split}
\end{equation}
Considering  $\sum\limits_{t=0}^{T-1}\sum\limits_{i=1}^{N}\mathbb{E}\{q_i(t)\}\geq\sum\limits_{t=0}^{T-1}\sum\limits_{i=1}^{N}\mathbb{E}\{E_i(\bm{\beta}^t)-\bar{E}_i\}$ yields the long-term energy deficit bound.

\section{Proof of Theorem 4}\label{prof_theorem_noncoop}
\begin{proof}
We first restate the problem introducing the variables $u_{ij}$ and $v_{ij}$, where $u_{ij}$ denotes SBS $i$'s workloads into SBS $j$ and $v_{ij}$ denotes SBS $i$'s workloads out of SBS $j$. Given our system model, we only have $v_{ij} \neq 0$ when $i=j$.

From the balance of the of the total traffic of SBS $i$ in the network, we have $\lambda_i=\sum_{j=1}^{N}u_{ij}=v_{ii}$.
The workload $\beta_{ij}$ offloaded from SBS $i$ to SBS $j$ can be written as
\begin{equation}
	\beta_{ij}=\left\{
	\begin{split}
	& u_{ij},&\text{if}~j\neq i \\
	&\phi_i-v_{ii},&\text{if}~j=i
	\end{split}\right.
\end{equation}
which can be rewritten as:
\begin{equation}
\beta_{ij}=\phi_{ij}+u_{ij}-v_{ij}~~
\text{where}
~~\phi_{ij}=\left\{
\begin{split}
&\phi_i,&\text{if}~j=i\\
& 0, &\text{if}~j\neq i
\end{split} \right.
\end{equation}

Using above equations, the $\textbf{P3}$ becomes
\begin{subequations}
	\begin{align}
	\min_{\bm{u}_{i\cdot},\bm{v}_{i\cdot}} C_i(\bm{u}_{i\cdot},\bm{v}_{i\cdot}) =  &\sum_{j=1}^{N}[\dfrac{V(\phi_{ij}+u_{ij}-v_{ij})}{\mu_{ij}-(\phi_{ij}+u_{ij}-v_{ij})}+\dfrac{V\sum_{j=1}^{N}v_{ij}}{g_{-i}-\tau\sum_{j=1}^{N}v_{ij}}]+\kappa q_i(\phi_i-\sum_{j=1}^{N}v_{ij})\\
	\text{s.t.}\qquad&-\sum_{j=1}^{N}u_{ij}+\sum_{j=1}^{N}v_{ij}=0\\
	&\phi_{ij}+u_{ij}-v_{ij}\geq 0\\
	& u_{ij}\geq 0\\
	& v_{ij}\geq 0
	\end{align}
\end{subequations}
Let $\alpha_i$, $\delta_j\leq 0$, $\psi_j\leq 0$, $\eta_j\leq 0$ denote the Lagrange multipliers. The Lagrangian is
\begin{align}
& L(\bm{u}_{i\cdot}, \bm{v}_{i\cdot}, \alpha_i,\bm{\delta}, \bm{\psi},\bm{\eta}) =C_i(\bm{u}_{i\cdot},\bm{v}_{i\cdot})+\alpha_i(\sum_{j=1}^{N}v_{ij}-\sum_{j=1}^{N}u_{ij})\nonumber\\
&\quad+\sum_{j=1}^{N}\delta_j(u_{ij}-v_{ij}+\phi_{ij})+\sum_{j=1}^{N}\psi_ju_{ij}+\sum_{j=1}^{N}\eta_jv_{ij}
\end{align}

The optimal solution satisfies the following Kuhn-Tucker conditions:
\begin{subequations}
	\begin{align}
	&\dfrac{\partial L}{\partial u_{ij}}=Vd_{ij}(\phi_{ij}+u_{ij}-v_{ij})-\alpha_i+\delta_j+\psi_j=0 \label{partial_u}\\
	&\dfrac{\partial L}{\partial v_{ij}}=-Vd_{ij}(\phi_{ij}+u_{ij}-v_{ij})+Vg_i(\lambda)
	-q_i\kappa+\alpha_i-\delta_j+\eta_j=0 \label{partial_v}\\
	&\dfrac{\partial L}{\partial \alpha_i}=-\sum_{j=1}^{N}u_{ij}+\sum_{j=1}^{N}v_{ij}=0 \label{sum_u_sum_v}\\
	&\phi_{ij}+u_{ij}-v_{ij}\geq 0,~\delta_j(\phi_{ij}+u_{ij}-v_{ij})=0,~\delta_j<0 \label{phi_u_v}\\
	& u_{ij}\geq 0,~\psi_ju_{ij}=0,~\psi_j<0 \label{u}\\
	& v_{ij}\geq 0,~\eta_jv_{ij}=0,~\eta_j<0 \label{v}
	\end{align}
\end{subequations}

Next, we find an equivalent form of (\ref{partial_u}) - (\ref{v}) in terms of $\bm{\beta}_{i}$. We consider two case:
\begin{itemize}
	\item \emph{Case I}: $\phi_{ij}+u_{ij}-v_{ij}=0$. We have $\beta_{ij}=0$
  \begin{itemize}[leftmargin=*]
	   \item \emph{Case I.1}: $j\neq i$. We have $\phi_{ij}=0$ and $v_{ij}=0$, it follows that $u_{ij}=0$. From (\ref{partial_u}), we have $Vd_{ij}(\beta_{ij})=\alpha_i-\delta_j-\psi_j$, which means $Vd_{ij}(\beta_{ij})>\alpha_i$. This corresponds to \emph{idle SBS}.
	   \item \emph{Case I.2}: $j=i$. We have $\phi_{ii}=\phi_i$, it follows that $v_{ii}=\phi_i$, (\ref{v}) implies $\eta_i=0$. From (\ref{partial_v}), we have $Vd_{ii}(\beta_{ii})=\alpha_i+Vg_i(\lambda)-\kappa q_i-\delta_j$. Then from (\ref{sum_u_sum_v}), we get $Vd_{ii}(\bm{\beta}_{i})+\kappa q_i>\alpha_i+Vg_i(\lambda)$. This corresponds to \emph{source SBS}.
  \end{itemize}	
	
	\item \emph{Case II}: $\phi_{ij}+u_{ij}-v_{ij}>0$. From (\ref{sum_u_sum_v}), we have $\delta_j=0$
    \begin{itemize}
 	\item \emph{Case II.1}: $v_{ij}>0$, we must have $j=i$. From (\ref{v}), we have $\eta_j=0$. (\ref{partial_v}) implies $Vd_{ii}(\beta_{ii})+\kappa q_i=\alpha_i+Vg_i(\lambda)$. This case corresponds to \emph{source SBS}.
 	\item \emph{Case II.2}: $v_{ij}=0$.
 	   \begin{itemize}
 	   	\item \emph{Case II.2.1}: $u_{ij}=0$. For $i\neq j$, it is equivalent to \emph{Case I.1}. We consider $j=i$. Then $\beta_{ii}=\phi_{i}$. From (\ref{partial_v}) and (\ref{u}), we have $Vd_{ii}(\beta_{ii})+\kappa q_i<\alpha_i+Vg_i(\lambda)$. This case corresponds to the \emph{neutral SBS}.
 	   	\item \emph{Case II.2.2}: $u_{ij}>0$, we must have $j\neq i$. Then $\beta_{ij}>\phi_{ij}=0$. From (\ref{u}), we have $\psi_j=0$. Substituting this in (\ref{partial_u}), we have
 	   	$Vd_{ij}(\beta_{ij})=\alpha_i$. This case corresponds to \emph{sink SBS} for SBS $i$.
 	   \end{itemize}
    \end{itemize}
\end{itemize}

Equation (\ref{sum_u_sum_v}) may be written in terms of $\beta_{ij}$ as
\begin{equation*}
\sum_{j\in\mathcal{S}_i}d_{ij}^{-1}(\dfrac{\alpha_i}{V})=\sum_{i\in\mathcal{R}}\left(\phi_i-[d_{ij}^{-1}(\dfrac{1}{V}(\alpha_i+Vg_i(\lambda)-q_i\kappa))]^+\right)
\end{equation*}
where $\alpha_i$ is the Lagrange multiplier.
\end{proof}

\section{Proof of Theorem \ref{Theorem: performance_OPEN_aut}}
The proof of Theorem \ref{Theorem: performance_OPEN_aut} is similar to that of Theorem \ref{Theorem_Online}. We reuse the assumption in \eqref{eq:assump}. Given the \emph{drift-plus-cost} bound in \eqref{drift_plus_cost} and the PoA bound in \eqref{max_PoA}, we have
\begin{align}\label{draft_plus_cost_NE}
&\Delta(\bm{q}(t))+V\sum_{i=1}^{N}\mathbb{E}\left\{D_i^t(\bm{\beta}^{\text{NE},t})|\bm{q}(t)\right\} \nonumber\\
&\leq B+\sum_{i=1}^{N}q_i(t)\mathbb{E}\left\{(E^t_i(\bm{\beta}^{\text{NE},t})-\bar{E}_i)|\bm{q}(t) \right\}+V\sum_{i=1}^{N}\mathbb{E}\left\{D_i^t(\bm{\beta}^{\text{NE},t})|\bm{q}(t)\right\} \nonumber\allowdisplaybreaks\\
&\leq B+\sum_{i=1}^{N}q_i(t)\mathbb{E}\left\{(\varrho^{\max}E^t_i(\bm{\beta}^{*,t})-\bar{E}_i)|\bm{q}(t) \right\}+\varrho^{\max}V\sum_{i=1}^{N}\mathbb{E}\left\{D_i^t(\bm{\beta}^{*,t})|\bm{q}(t)\right\}  \nonumber\\
&\leq B+\sum_{i=1}^{N}q_i(t)\mathbb{E}\left\{(\varrho^{\max}E^t_i(\bm{\beta}^{L,t})-\varrho^{\max}\bar{E}_i)|\bm{q}(t)\right\}  +\sum_{i=1}^{N}q_i(t)\mathbb{E}\left\{(\varrho^{\max}-1)\bar{E}_i)|\bm{q}(t)\right\}\nonumber\\
&\qquad\qquad\qquad +\varrho^{\max}V\sum_{i=1}^{N}\mathbb{E}\left\{D_i^t(\bm{\beta}^{L,t})|\bm{q}(t)\right\}\allowdisplaybreaks \nonumber\\
&\leq B + \left((\varrho^{\max}-1)\bar{E}^{\max}-\varrho^{\max} \epsilon\right) \sum_{i=1}^{N}q_i(t)+\varrho^{\max}V\Psi(\epsilon)\nonumber
\end{align}
Letting $\epsilon=\frac{\varrho^{\max}-1}{\varrho^{\max}}\bar{E}^\max$ in \eqref{eq:assump}, we have
\begin{align*}
	&\Delta(\bm{q}(t))+V\sum_{i=1}^{N}\mathbb{E}\left\{D_i^t(\bm{\beta}^{\text{NE},t})|\bm{q}(t)\right\} \leq B +\varrho^{\max}V\Psi(\frac{\varrho^{\max}-1}{\varrho^{\max}}\bar{E}^\max)
\end{align*}
Summing the inequality over $t\in\{0,1,\dots,T-1\}$ and then dividing the result by $T$, we have:
\begin{align}
&\dfrac{1}{T}\mathbb{E}\{L(\bm{q}(t))-L(\bm{q}(0))\}+\dfrac{V}{T}\sum_{t=0}^{T-1}\sum_{i=1}^{N}\mathbb{E}\left\{D_i^t(\bm{\beta}^{\text{NE},t})\right\}\leq B+\varrho^{\max}V\Psi(\frac{\varrho^{\max}-1}{\varrho^{\max}}\bar{E}^\max)
\end{align}
Rearranging the terms and considering the fact that $L(\bm{q}(t))\geq 0$ and $L(\bm{q}(0))=0$ yields the bound for the long-term system delay cost for OPEN-Autonomous.

To obtain the long-term energy-deficit bound for OPEN-Autonomous, we directly sum \eqref{draft_plus_cost_NE} over $t\in\{0,1,\dots,T-1\}$ and rearrange terms:
\begin{equation*}
\begin{split}
&\dfrac{1}{T}\sum_{t=0}^{T-1}\sum_{i=1}^{N}\mathbb{E}\{q_i(t)\}\\
&\leq \dfrac{1}{\varrho^{\max} \epsilon-(\varrho^{\max}-1)\bar{E}^{\max}}\cdot\\
&\qquad\left(B+V(\varrho^{\max}\tilde{\Psi}(\epsilon)-\dfrac{1}{T}\sum\limits_{t=0}^{T-1}\sum\limits_{i=1}^{N}\mathbb{E}\left\{D_i^t(\bm{\beta}^{\text{NE},t})\right\}\right)\\
&\leq \dfrac{1}{\varrho^{\max} \epsilon-(\varrho^{\max}-1)\bar{E}^{\max}}\left(B+V(\varrho^{\max}D^{\max}_{\text{sys}}-D^{\text{opt}}_{\text{sys}})\right)
\end{split}
\end{equation*}
Considering  $\sum\limits_{t=0}^{T-1}\sum\limits_{i=1}^{N}\mathbb{E}\{q_i(t)\}\geq\sum\limits_{t=0}^{T-1}\sum\limits_{i=1}^{N}\mathbb{E}\{E_i(\bm{\beta}^{\text{NE},t})-\bar{E}_i\}$ yields the energy deficit bound.
\end{document}